%% file: main.tex
\documentclass[11pt]{article}

\usepackage[lmargin=1.4in,rmargin=1.4in,bottom=1.3in,top=1.3in,twoside=False]{geometry}

\usepackage[absolute]{textpos}
\usepackage{enumerate}
\usepackage[shortlabels,inline]{enumitem}

\usepackage{relsize,xspace}
 \usepackage{xcolor}
 \usepackage{mathtools}
 \usepackage{comment}
\usepackage{microtype}
\usepackage{amsmath}
\usepackage{amssymb}
\usepackage{amsfonts}
\usepackage{mathrsfs}
\usepackage{bm}
\usepackage{tikz}
\usepackage{refcount}
\usepackage{wrapfig}

\usepackage{marginnote}
\usepackage{todonotes}

\newcommand{\Uu}{\mathcal{U}}

\newcommand{\Cc}{\mathscr{C}}

\newcommand{\N}{\mathbb{N}}

\renewcommand{\phi}{\varphi}
\renewcommand{\epsilon}{\varepsilon}
\newcommand{\strA}{\str{A}}

\renewcommand{\mid}{~:~}

\newcommand{\from}{\colon}

\renewcommand{\leq}{\leqslant}
\renewcommand{\geq}{\geqslant}
\renewcommand{\le}{\leqslant}
\renewcommand{\ge}{\geqslant}

\newcommand{\set}[1]{\{#1\}}
\newcommand{\setof}[2]{\left\{#1 \,\mid\, #2 \right\}}
\renewcommand{\subset}{\subseteq}
\newcommand{\DKT}{Dvo\v r\'ak, Kr\'al', and Thomas\xspace}
\newcommand{\pow}[1]{{\mathscr P}(#1)}
\newcommand{\str}[1]{\mathbf #1}
\newcommand{\class}[1]{\mathscr #1}
\newcommand{\interp}[1]{\mathsf #1}

\renewcommand{\cal}[1]{\mathcal{#1}}
\newcommand{\dom}{\mathsf{dom}}
\newcommand{\partto}{\rightharpoonup}

\newcommand{\wh}[1]{\widehat{#1}}

\newcommand{\CCC}{\class{C}}

\newcommand{\DDD}{\class{D}}
\newcommand{\PPP}{\class{P}}
\newcommand{\FFF}{\class{F}}

\newcommand{\TTT}{\class{T}}

\newcommand{\LLL}{\mathcal{L}}

\usepackage{ragged2e}

\usepackage{cite}

\usepackage{graphicx,algorithmicx, algpseudocode, algorithm}
\usepackage{enumerate}
\usepackage{tikz}
\usetikzlibrary{shapes}

 \usepackage{xcolor}
 \usepackage{mathtools}
\usepackage{microtype}
\usepackage{amsfonts}
\usepackage{comment}
\usepackage[english]{babel}

\definecolor{blue}{rgb}{0.1,0.2,0.5}
\definecolor{brown}{rgb}{0.6,0.6,0.2}

\usepackage[amsmath,thmmarks,hyperref]{ntheorem}
\usepackage{cleveref}

\crefformat{page}{#2page~#1#3}%
\Crefformat{page}{#2Page~#1#3}%
\crefformat{equation}{#2(#1)#3}%
\Crefformat{equation}{#2(#1)#3}%
\crefformat{figure}{#2Figure~#1#3}%
\Crefformat{figure}{#2Figure~#1#3}%
\crefformat{section}{#2Section~#1#3}
\Crefformat{section}{#2Section~#1#3}
\crefformat{chapter}{#2Chapter~#1#3}
\Crefformat{chapter}{#2Chapter~#1#3}
\crefformat{chapter*}{#2Chapter~#1#3}
\Crefformat{chapter*}{#2Chapter~#1#3}
\crefformat{part}{#2Part~#1#3}
\Crefformat{part}{#2Part~#1#3}
\crefformat{enumi}{#2(#1)#3}
\Crefformat{enumi}{#2(#1)#3}

\usepackage{enumerate}

\usepackage{latexsym}

\theoremnumbering{arabic}
\theoremstyle{plain}
\theoremsymbol{}
\theorembodyfont{\itshape}
\theoremheaderfont{\normalfont\bfseries}
\theoremseparator{.}

\newtheorem{theorem}{Theorem}
\crefformat{theorem}{#2Theorem~#1#3}
\Crefformat{theorem}{#2Theorem~#1#3}

\newcommand{\newtheoremwithcrefformat}[2]{%
  \newtheorem{#1}[theorem]{#2}%
  \crefformat{#1}{##2\MakeUppercase#1~##1##3}%
  \Crefformat{#1}{##2\MakeUppercase#1~##1##3}%
}
\newcommand{\newseptheoremwithcrefformat}[2]{%
  \newtheorem{#1}{#2}%
  \crefformat{#1}{##2\MakeUppercase#1~##1##3}%
  \Crefformat{#1}{##2\MakeUppercase#1~##1##3}%
}

\newtheoremwithcrefformat{lemma}{Lemma}
\newtheoremwithcrefformat{proposition}{Proposition}
\newtheoremwithcrefformat{observation}{Observation}
\newtheoremwithcrefformat{conjecture}{Conjecture}
\newtheoremwithcrefformat{corollary}{Corollary}
\newseptheoremwithcrefformat{claim}{Claim}
\theorembodyfont{\upshape}
\newtheoremwithcrefformat{example}{Example}
\newtheoremwithcrefformat{remark}{Remark}
\newtheoremwithcrefformat{definition}{Definition}

\theoremstyle{nonumberplain}
\theoremheaderfont{\scshape}
\theorembodyfont{\normalfont}
\theoremsymbol{\ensuremath{\square}}
\newtheorem{proof}{Proof}

\theoremsymbol{\ensuremath{\lrcorner}}
\newtheorem{clproof}{Proof}

\def\cqedsymbol{\ifmmode$\lrcorner$\else{\unskip\nobreak\hfil
\penalty50\hskip1em\null\nobreak\hfil$\lrcorner$
\parfillskip=0pt\finalhyphendemerits=0\endgraf}\fi}

\newcounter{tmp}

\usepackage{authblk}

\title{First-order interpretations  of  bounded expansion classes
\footnote{J. Ne{\v s}et{\v r}il and P. Ossona~de~Mendez are 
supported by CE-ITI P202/12/G061 of GACR and European 
Associated Laboratory (LEA STRUCO). 
J.\ Gajarsk\'y and S.\ Kreutzer are supported by the
European Research Council (ERC) under the European Union's Horizon
2020 research and innovation programme (ERC Consolidator Grant
DISTRUCT, grant agreement No 648527).
M.\ Pilipczuk and S.\ Siebertz are supported by the National 
Science Centre of Poland (NCN) via POLONEZ grant agreement
UMO-2015/19/P/ST6/03998, which has received funding from 
the European Union's Horizon 2020 research and innovation 
programme (Marie Sk\l odowska-Curie grant agreement No.\ 665778).
\newline
Sz.\ Toru{\'n}czyk is supported by the NCN  grant 
2016/21/D/ST6/01485.}}

\author[1]{Jakub~Gajarsk\'y}
\author[1]{Stephan~Kreutzer}
\author[2]{Jaroslav~Ne{\v s}et{\v r}il}
\author[2,3]{Patrice~Ossona~de~Mendez}
\author[4]{Micha{\l}~Pilipczuk}
\author[4]{Sebastian~Siebertz}
\author[4]{Szymon~Toru\'nczyk}

\affil[1]{Technical University Berlin, Germany, \texttt{\{jakub.gajarsky,stephan.kreutzer\}@tu-berlin.de}}
\affil[2]{Charles University, Prague, Czech Republic, \texttt{\{nesetril,patrice\}@kam.mff.cuni.cz}}
\affil[3]{CAMS (CNRS, UMR 8557), Paris, France, \texttt{pom@ehess.fr}}
\affil[4]{University of Warsaw, Poland, \texttt{\{michal.pilipczuk,siebertz,szymtor\}@mimuw.edu.pl}}
\date{}

\begin{document}
\maketitle
\begin{abstract}
  The notion of bounded expansion captures uniform sparsity of graph
  classes and renders various algorithmic problems that are hard in
  general tractable. In particular, the model-checking problem for
  first-order logic is fixed-parameter tractable over such graph
  classes. With the aim of generalizing such results to dense graphs,
  we introduce classes of graphs with \emph{structurally bounded
    expansion}, defined as first-order interpretations of classes of
  bounded expansion. As a first step towards their algorithmic
  treatment, we provide their characterization analogous to the
  characterization of classes of bounded expansion via low treedepth
  decompositions, replacing treedepth by its dense analogue called
  shrubdepth.
\end{abstract}

\begin{textblock}{5}(12, 13.8)
\includegraphics[width=38px]{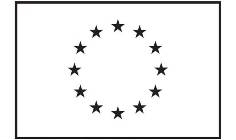}%
\end{textblock}

\input{intro_full}
\input{prelims_full}
\input{sbe_covers_full}
\input{outline}

\input{definability_full}
\input{qe-covers}

\input{main-proof_full}
\input{effectivity_short}
\input{conclusion_full}

\bibliographystyle{abbrv}
\bibliography{biblio-seb,bib-seb,stephan} 

\newpage

\appendix
\makeatletter
\edef\thetheorem{\expandafter\noexpand\thesection\@thmcountersep\@thmcounter{theorem}}
\makeatother
\input{app-normal}

\input{bicliques}
\input{app-sbe_covers}

\input{locality}
\input{Treducible}

\input{qe-new}

\end{document}

%% file: intro_full.tex
\section{Introduction}
The interplay of methods from logic and graph theory has led to many
important results in theoretical computer science, notably in
algorithmics and complexity theory.  The combination of logic and
algorithmic graph theory is particularly fruitful in the area of
\emph{algorithmic meta-theorems}.  Algorithmic meta-theorems are
results of the form: \emph{every computational problem definable in a
logic \(\LLL\) can be solved efficiently on any class of structures satisfying a property~\(\PPP\).} 
In other words, these theorems show that the \emph{model-checking
problem} for the logic \(\LLL\) on any class \(\CCC\) satisfying 
\(\PPP\) can be solved efficiently, where \emph{efficiency} usually 
means \emph{fixed-parameter tractability}.

The archetypal example of an algorithmic meta-theorem is Courcelle's
theorem~\cite{Courcelle1,Courcelle2}, which states that model-checking
a formula~$\phi$ of \emph{monadic second-order logic} can be solved in
time $f(\phi)\cdot n$ on any graph with $n$ vertices which comes from
a fixed class of graphs of \emph{bounded treewidth}, for some
computable function $f$. 
Seese~\cite{Seese1996} proved an analogue of Courcelle's result for
the model-checking problem of first-order logic on any class of graphs
of bounded degree. Following this result, the complexity of
first-order model-checking on specific classes of graphs has been
studied extensively in the literature. See
e.g.~\cite{GroheK11,DKT2,GroheKS17,Kazana2013,Dawar2007a,Flum2001,Frick2001,Seese1996,
  eickmeyer2017map,kazana2013enumeration,segoufin2011first,grohe2001generalized,durand2007first,durand2014enumerating,ganian2013fo,hlineny2017}.
One of the main goals of this line of research is to find a structural
property \(\PPP\)
which precisely defines those graph classes \(\CCC\)
for which model checking of first-order logic is tractable.

So far, research on algorithmic meta-theorems has focused
predominantly on \emph{sparse} classes of graphs, such as classes of
bounded \emph{treewidth}, \emph{excluding a minor} or which have
\emph{bounded expansion} or are \emph{nowhere dense}.
The concepts of \emph{bounded expansion} and \emph{nowhere denseness}
were introduced by Ne{\v s}et{\v r}il and Ossona~de~Mendez with the
goal of capturing the intuitive notion of \emph{sparseness}. See
\cite{Sparsity} for an extensive cover of these notions. The large
number of equivalent ways in which they can be defined using either
notions from combinatorics, theoretical computer science or logic,
indicate that these two concepts capture some very natural limits of
``well-behavedness'' and algorithmic tractability.  For instance,
Grohe et al.~\cite{GroheKS17} proved that if \(\CCC\)
is a class of graphs closed under taking subgraphs then model checking
first-order logic on $\CCC$ is tractable if, and only if, \(\CCC\)
is nowhere dense (the lower bound was proved in~\cite{DKT2}). As far
as algorithmic meta-theorems for fixed-parameter tractability of
first-order model-checking are concerned, this result completely
solves the case for graph classes which are \emph{closed under taking
  subgraphs}, which is a reasonable requirement for sparse but not for
dense graph classes.

Consequently, research in this area has shifted towards studying the
dense case, which is much less understood.
While there are several examples of algorithmic meta-theorems on dense
classes, such as for monadic second-order logic on classes of bounded
\emph{cliquewidth} \cite{Courcelle2000a} or for first-order logic on
\emph{interval graphs}, \emph{partial orders}, classes of bounded
\emph{shrubdepth} and other classes, see
e.g.~\cite{ganian2013fo,gajarsky2015fo,Ganian2012,GajarskyHOLR16}, a
general theory of meta-theorems for dense classes is still
missing. Moreover, unlike the sparse case, there is no canonical
hierarchy of dense graph classes similar to the sparse case which
could guide research on algorithmic meta-theorems in the dense world.

Hence, the main research challenge for dense model-checking is not
only to prove tractability results and to develop the necessary
logical and algorithmic tools. It is at least as important to define
and analyze promising candidates for 
``structurally simple'' classes of graph classes which are not
necessarily sparse. This is the main motivation for the research in
this paper.  Since bounded expansion and nowhere denseness form the
limits for tractability of certain problems in the sparse case, any
extension of the theory should provide notions which collapse to
bounded expansion or nowhere denseness, under the additional
assumption that the classes are closed under taking subgraphs.
Therefore, a natural way of seeking such notions is to base them on
the existing notions of bounded expansion or nowhere denseness.

In this paper, we take bounded expansion classes as a starting point
and study two different ways of generalizing them towards dense graph
classes preserving their good properties.  In particular, we define
and analyze classes of graphs obtained from bounded expansion classes
by means of first-order interpretations and classes of graphs obtained
by generalizing another, more combinatorial characterization of
bounded expansion in terms of low treedepth colorings into the dense
world. Our main structural result shows that these two very different
ways of generalizing bounded expansion into the dense setting lead to
the same classes of graphs. This is explained in greater detail below.

\medskip
\noindent
\textbf{Interpretations and transductions.} One possible way of
constructing ``well-behaved'' and ``structurally simple'' classes of
graphs is to use logical \emph{interpretations}, or the related
concept of \emph{transductions} studied in formal language and
automata theory. For our purpose, transductions are more convenient
and we will use them in this paper.  Intuitively, a
\emph{transduction} is a logically defined operation which takes a
structure as input and nondeterministically produces as output a
target structure.
In this paper we use \emph{first-order} transductions, which involve
first-order formulas (see \cref{sec:preliminaries} for details).  Two
examples of such transductions are graph complementation, and the
squaring operation which, given a graph $G$, adds an edge between
every pair of vertices at distance $2$ from each other.

We postulate that if we start with a ``structurally simple'' class
\(\CCC\)
of graphs, e.g.~a class of bounded expansion or a nowhere dense class,
and then study the graph classes~\(\DDD\)
which can be obtained from \(\CCC\)
by first-order transductions, then the resulting classes should still
have a simple structure and thus be well-behaved algorithmically as
well as in terms of logic. In other words,~the resulting classes are
interesting graph classes with good algorithmic and logical
properties, and which are certainly not sparse in general.  For
instance, a useful feature of transductions is that they provide a
canonical way of reducing model-checking problems from the generated
classes \(\DDD\)
to the original class~\(\CCC\),
provided that given a graph $H\in \DDD$, we can effectively compute
some graph $G\in\CCC$ that is mapped to $H$ by the transduction.  In
general, this is a hard problem, requiring a \emph{combinatorial}
understanding of the structure of the resulting classes $\DDD$.


The above principle has so far been successfully applied in the
setting of graph classes of bounded treewidth and monadic second-order
transductions: it was shown by Courcelle, Makowsky and
Rotics~\cite{courcelle2000linear} that transductions of classes of
bounded treewidth can be combinatorially characterized as classes of
bounded cliquewidth. This, combined with Oum's
result~\cite{oum2008approximating} gives a fixed-parameter algorithm
for model-checking monadic second-order logic on classes of bounded
cliquewidth.  More recently, the same principle, but for first-order
logic, has been applied to graphs of bounded
degree~\cite{GajarskyHOLR16}, leading to a combinatorial
characterization of first-order transductions of such classes, and to
a model-checking algorithm.

Applying our postulate to bounded expansion classes yields the central
notion of this paper: a class of graphs has \emph{structurally bounded
  expansion} if it is the image of a class of bounded expansion under
some fixed first-order transduction.  This paper is a step towards a
combinatorial, algorithmic, and logical understanding of such
graph~classes.

\bigskip
\noindent\textbf{Low Shrubdepth Covers. }
The method of transductions is one way of constructing complex graphs
out of simple graphs.  A more combinatorial approach is the method of
\emph{decompositions} (or \emph{colorings})~\cite{Sparsity}, which we
reformulate below in terms of \emph{covers}.  This method can be used
to provide a characterization of bounded expansion classes in terms of
very simple graph classes, namely classes of \emph{bounded
  treedepth}. A class of graphs has bounded treedepth if there is a
bound on the length of simple paths in the graphs in the class (see
\cref{sec:preliminaries} for a different but equivalent definition). A
class \(\CCC\)
has \emph{low treedepth covers} if for every number \(p \in \N\)
there is a number \(N\)
and a class of bounded treedepth $\cal T$ such that for every
$G\in\CCC$, the vertex set $V(G)$ can be covered by \(N\)
sets $U_1,\ldots,U_N$ so that every set $X\subset V(G)$ of at most $p$
vertices is contained in some $U_i$, and for each $i=1,\ldots,N$, the
subgraph of $G$ induced by $U_i$ belongs to $\cal T$.  A consequence
of a result by Ne{\v s}et{\v r}il and Ossona~de~Mendez~\cite{POMNI} on
a related notion of \emph{low treedepth colorings} is that a graph
class has bounded expansion if, and only if, it has low treedepth
covers.

The decomposition method allows to lift algorithmic, logical, and
structural properties from classes of bounded treedepth to classes of
bounded expansion. For instance, this was used to show tractability of
first-order model-checking on bounded expansion
classes~\cite{dkt,Grohe2011}.

An analogue of treedepth in the dense world is the concept of
\emph{shrubdepth}, introduced in~\cite{Ganian2012}. Shrubdepth shares
many of the good algorithmic and logical properties of treedepth.
This notion is defined combinatorially, in the spirit of the
definition of cliquewidth, but can be also characterized by logical
means, as first-order transductions of classes of bounded treedepth.
Applying the method of decompositions to the notion of shrubdepth
leads to the following definition.  A class \(\CCC\)
of graphs has \emph{low shrubdepth covers} if for every number
\(p \in \N\)
there is a number~\(N\)
and a class $\cal B$ of bounded shrubdepth such that for every
$G\in\CCC$, there is a \emph{$p$-cover} of $G$ consisting of \(N\)
sets $U_1,\ldots,U_N\subset V(G)$, so that every set $X\subset V(G)$
of at most $p$ vertices is contained in some $U_i$ and for each
$i=1,\ldots,N$, the subgraph of $G$ induced by $U_i$ belongs
to~$\cal B$.  Shrubdepth properly generalizes treedepth and
consequently classes admitting low shrubdepth covers properly extend
bounded expansion classes.

It was observed earlier~\cite{KwonPS17} that for every fixed $r\in\N$
and every class $\Cc$ of bounded expansion, the class of $r$th power
graphs $G^r$ of graphs from $\Cc$ (the $r$th power of a graph is a
simple first-order transduction) admits low shrubdepth colorings.

\bigskip\noindent\textbf{Our contributions.}  Our main result,
\cref{thm:main}, states that the two notions introduced above are the
same: a class of graphs $\class{C}$ has structurally bounded expansion
if, and only if, it has bounded shrubdepth covers.  That is,
transductions of classes of bounded expansion are the same as classes
with low shrubdepth covers (cf. Figure~\ref{fig:main}).  This gives a
combinatorial characterization of structurally bounded expansion
classes, which is an important step towards their algorithmic
treatment.

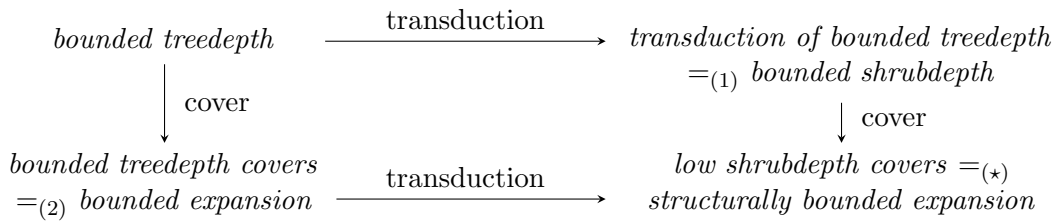
\begin{figure*}
  \centering \input{classes-interpretation-stephan}
  \caption{ The nodes in the diagram depict properties of graph
    classes, and the arrows depict operations on properties of graph
    classes.  Equality (1) is by~\cite{Ganian2012}. Equality (2) is
    by~\cite{POMNI}. Equality~($\star$) is the main result of this
    paper, \cref{thm:main}.  }
  \label{fig:main}
\end{figure*}

One of the key ingredients of our proof is a quantifier-elimination
result (\cref{thm:qe-lsc}) for transductions on classes of
structurally bounded expansion. This result strengthens in several
ways similar results for bounded expansion classes due to
\DKT~\cite{dkt}, Grohe and Kreutzer~\cite{Grohe2011} and Kazana and
Segoufin~\cite{kazana2013enumeration}.  Our assumption is more
general, as they assume that $\CCC$ has bounded expansion, and here
$\CCC$ is only required to have low shrubdepth covers. Also, our
conclusion is stronger, as their results provide quantifier-free
formulas involving some unary functions and unary predicates which are
computable algorithmically, whereas our result shows that these
functions can be defined using very restricted transductions.
Quantifier-elimination results of this type proved to be useful for
the model-checking problem on bounded expansion
classes~\cite{dkt,Grohe2011,kazana2013enumeration}, and this is also
the case here.

As explained earlier, the transduction method allows to reduce the
model-checking problem to the problem of finding inverse images under
transductions, which is a hard problem in general and depends very
much on the specific transduction.  On the other hand, as we show, the
cover method allows to reduce the model-checking problem for classes
with low shrubdepth covers to the problem of computing a bounded
shrubdepth cover of a given graph.  In fact, as a consequence of our
proof, in \cref{thm:algo} we show that it is enough to compute a
$2$-cover of a given graph $G$ from a structurally bounded expansion
class, in order to obtain an algorithm for the model-checking problem
for such classes.  We conjecture that such an algorithm exists and
that therefore first-order model-checking is fixed-parameter tractable
on any class of graphs of structurally bounded expansion. We leave
this problem for future work.

\bigskip\noindent\textbf{Organization.} In \cref{sec:preliminaries} we
collect basic facts about logic, transductions, treedepth, shrubdepth
and the notion of bounded expansion. In \cref{sec:sbe} we provide the
formal definitions of structurally bounded expansion classes and
classes with low shrubdepth covers, and state the main results and
their proofs using lemmas which are proved in the following three
sections.
We consider algorithmic aspects in \cref{sec:alg} and conclude in
\cref{sec:conclusion}.  We aim to present an easy to follow proof of
our main result. For this reason, we present proofs of the key lemmas
in the main body of the paper, while rather technical results that
disturb the flow of ideas are presented in full detail in the
appendix.


%% file: classes-interpretation-stephan.tex
\begin{tikzpicture}

\node (btd) at (0,0) {\parbox{4cm}{\centering\textit{bounded treedepth}}};

\node (ibtd) at (0,-2.1) {\parbox{4.3cm}{\centering\textit{bounded treedepth
      covers}\\\(=_{(2)}\) \textit{bounded expansion}
      \\\phantom{x}
      }};

\node (cbtd) at (9,0) {\parbox{6cm}{\centering\textit{\phantom{I}\\ transduction of bounded
      treedepth}\\\(=_{(1)}\) \textit{bounded shrubdepth}}};

\node (ibe) at (9,-2.1) {\parbox{6cm}{\centering\textit{low shrubdepth covers}
    \(=_{(\star)}\)\\\textit{structurally bounded expansion\\{\phantom{x}}}}};

\draw[->, >=stealth] (btd) -- (ibtd);
\draw[->, >=stealth] (btd) -- (cbtd);
\draw[->, >=stealth] (cbtd) -- (ibe);
\draw[->, >=stealth] (ibtd) -- (ibe);

\node at (0.7,-0.85) {\rotatebox[origin=c]{0}{cover}};

\node at (9.7,-1.05) {\rotatebox[origin=c]{0}{cover}};

\node at (4,0.3) {transduction};

\node at (4,-1.8) {transduction};
\end{tikzpicture}


%% file: prelims_full.tex
\section{Preliminaries}\label{sec:preliminaries}
\paragraph*{Basic notation.}  We use standard graph notation. All
graphs considered in this paper are undirected, finite, and simple;
that is, we do not allow loops or multiple edges with the same pair of
endpoints.  We follow the convention that the composition of an empty
sequence of (partial) functions is the identity function.  For an
integer $k$, we denote $[k]=\{1,\ldots,k\}$.

\subsection{Structures, logic, and transductions}\label{sec:prelims-logic}

\paragraph*{Structures and logic.}  A {\em{signature}} $\Sigma$ is a
finite set of relation symbols, each with prescribed arity that is a
non-negative integer, and unary function symbols.  A
\mbox{{\em{structure}}}~$\str A$ over $\Sigma$ consists of a finite
universe $V(\str A)$ and {\em{interpretations}} of symbols from the
signature: each relation symbol $R\in \Sigma$, say of arity $k$, is
interpreted as a $k$-ary relation $R^{\str A}\subseteq V(\str A)^k$,
whereas each function symbol $f$ is interpreted as a {\em{partial}}
function $f^{\str A}\colon V(\str A)\partto V(\str A)$. We drop the
superscipt when the structure is clear from the context, thus
identifying each symbol with its interpretation.  If \(\strA\)
is a structure and \(X\subseteq V(\strA)\)
then we define the \emph{substructure} of \(\strA\)
induced by \(X\)
in the usual way except that a unary function \(f(x)\)
in \(\strA\)
becomes undefined on all \(x \in X\)
for which \(f(x) \not\in X\).
The {\em{Gaifman graph}} of a structure $\str A$ is the graph with
vertex set $V(\str A)$ where two elements $u,v\in \str A$ are adjacent
if and only if either $u$ and $v$ appear together in some tuple in
some relation in $\str A$, or $f(u)=v$ or $f(v)=u$ for some partial
function $f$ in $\str A$.

For a signature $\Sigma$, we consider standard
first-order logic over $\Sigma$.  Let us clarify the usage of function
symbols.  A {\em{term}} $\tau(x)$ is a finite composition of function
symbols applied to a variable $x$.  In a structure $\str A$, given an
evaluation of $x$, the term $\tau(x)$ either evaluates to some element
of $\str A$ in the natural sense, or is {\em{undefined}} if during the
evaluation we encounter an element that does not belong to the domain
of the function that is to be applied next.
In first order logic over $\Sigma$ we allow usage of atomic formulas of the following form:
\begin{itemize}
\item $R(\tau_1(x_1),\ldots,\tau_k(x_k))$ for a relation symbol $R$ of
  arity $k$, terms $\tau_1,\ldots,\tau_k$, and variables
  $x_1,\ldots,x_k$;
\item $\tau_1(x_1)=\tau_2(x_2)$ for terms $\tau_1,\tau_2$ and variables $x_1,x_2$; and
\item $\dom_f(\tau(x))$ for term $\tau$ and variable $x$.
\end{itemize}
Here, the predicate $\dom_f(\tau(x))$ checks whether $\tau(x)$ belongs
to the domain of $f$.  The semantics are defined as usual, however an
atomic formula is false if any of the terms involved is undefined.
Based on these atomic formulas, the syntax and semantics of first
order logic is defined in the expected way.

\paragraph*{Graphs, colored graphs and trees.}  Graphs can be viewed
as finite structures over the signature consisting of a binary
relation symbol~$E$, interpreted as the edge relation, in the usual
way.  For a finite label set $\Lambda$, by a {\em{$\Lambda$-colored}}
graph we mean a graph enriched by a unary predicate~\(U_\lambda\)
for every $\lambda\in \Lambda$.  We will follow the convention that if
$\CCC$ is a class of colored graphs, then we implicitly assume that
all graphs in \(\CCC\)
are over the same fixed finite signature. A rooted forest is an
acyclic graph $F$ together with a unary predicate $R\subset V(F)$
selecting one root in each connected component of $F$. A tree is a
connected forest.  The {\em{depth}} of a node $x$ in a rooted forest
$F$ is the distance between $x$ and the root in the connected
component of $x$ in $F$. The depth of a forest is the largest depth of
any of its nodes.  The {\em{least common ancestor}} of nodes $x$ and
$y$ in a rooted tree is the common ancestor of $x$ and $y$ that has
the largest depth.

\paragraph*{Transductions.}  We now define the notion of transduction
used in the sequel. A \emph{transduction} is a special type of
first-order interpretation with set parameters, which we see here
(from a computational point of view) as a nondeterministic operation
that maps input structures to output structures.  Transductions are
defined as compositions of \emph{atomic operations} listed below.

An \textbf{extension} operation is parameterized by a first-order
formula $\phi(x_1,\ldots,x_k)$ and a relation symbol $R$.  Given an
input structure $\str A$, it outputs the structure $\str A$ extended
by the relation $R$ interpreted as the set of $k$-tuples of elements
satisfying $\phi$ in~$\str A$.  A \textbf{restriction} operation is
parameterized by a unary formula $\psi(x)$. Applied to a
structure~$\str A$ it outputs the substructure of~$\str A$ induced by
all elements satisfying \(\psi\).
A \textbf{reduct} operation is parameterized by a relation symbol
\(R\), and results in removing the relation~\(R\)
from the input structure.  \textbf{Copying} is an operation which,
given a structure~$\str A$ outputs a disjoint union of two copies of
$\str A$ extended with a new unary predicate which marks the newly
created vertices, and a symmetric binary relation which connects each
vertex with its copy.  A \textbf{function extension} operation is
parameterized by a binary formula $\varphi(x,y)$ and a function symbol
$f$, and extends a given input structure by a partial function $f$
defined as follows: $f(x)=y$ if $y$ is the unique vertex such that
$\varphi(x,y)$ holds.  Note that if there is no such $y$ or more than
one such $y$, then $f(x)$ is undefined.  Finally, suppose $\sigma$ is
function that maps each structure $\str A$ to a nonempty family
$\sigma(\str A)$ of subsets of its universe.  A \textbf{unary lift}
operation, parameterized by $\sigma$, takes as input a
structure~$\str A$ and outputs the structure~$\str A$ enriched by a
unary predicate~$X$ interpreted by a nondeterministically chosen set
$U\in\sigma(\str A)$.

We remark that function extension operations can be simulated by
extension operations, defining the graphs of the functions in the
obvious way.  They are, however, useful as a means of extending the
expressive power of transductions in which only quantifier-free
formulas are allowed, as defined below.

\emph{Transductions} are defined inductively: every atomic
transduction is a transduction, and the composition of two
transductions $\interp I$ and $\interp J$ is the transduction
$\interp I;\interp J$ that, given a structure $\str A$, first applies
$\interp I$ to $\strA$ and then $\interp J$ to the output
$\interp I(\strA)$.  A transduction is {\em{deterministic}} if it does
not use unary lifts. In this case, for every input structure there is
exactly one output structure.  A transduction is \emph{almost
  quantifier-free} if all formulas that parameterize atomic operations
comprising it are quantifier-free\footnote{We use the adverb
  ``almost'' to indicate that such transductions still can access
  elements that are not among its free variables via functions.}, and
is \emph{deterministic almost quantifier-free} if it additionally does
not use unary lifts.

\pagebreak 

If $\CCC$ is a class of structures, we write $\interp I(\CCC)$ for the
class which contains all possible outputs~$\interp I(\strA)$ for
$\strA\in \CCC$.  We say that two transductions~$\interp I$ and
$\interp J$ are \emph{equivalent} on a class~$\CCC$ of structures if
every possible output of $\interp I(\str A)$ is also a possible output
of $\interp J(\str A)$, and vice versa, for every $\str A \in \CCC$.

\medskip It may happen that an atomic operation $\interp I$ is
undefined for a given input structure $\str A$.  For example, for an
extension operation parametrized by a first order formula $\phi$ using
a relation symbol $R$, if the input structure $\str A$ does not carry
the symbol $R$, then $\interp I(\str A)$ is undefined according to the
above definition.  This will never occur in our constructions.
However, for completeness, we may define $\interp I(\str A)$ as a
fixed structure $\bot$ in such situations.

When considering a composition of atomic operations, we avoid
overriding symbols by later operations, i.e., we always assume that
subsequent atomic operations create relation symbols which are
distinct from previously created relations symbols and also from
symbols in the original signature.  Since every transduction
$\interp I$ is a composition of finitely many atomic operations, the
result of $\interp I$ applied to a structure over a finite signature
$\Sigma$ will be again a structure over a finite signature $\Gamma$,
which depends on $\Sigma$ and~$\interp I$ only (unless the result is
undefined).

\medskip

\begin{example}\label{ex:transduction}
  Let $\CCC$ be the class of rooted forests of depth at most $d$, for
  some fixed $d\in\N$.  We describe an almost quantifier-free
  transduction which defines the \emph{parent function} in $\CCC$.
  First, using unary lifts introduce $d+1$ unary predicates
  $D_0,...,D_d$, where~$D_i$ marks the vertices of the input tree
  which are at distance $i$ from a root. Next, using a function
  extension, define a partial function $f$ which maps a vertex $v$ in
  the input tree to its parent, or is undefined in case of a
  root. This can be done by a quantifier-free formula, which selects
  those pairs $x,y$ such that $x$ and $y$ are adjacent and $D_i(x)$
  implies $D_{i-1}(y)$.
\end{example}

It will sometimes be convenient to work with the encoding of
bounded-depth trees and forests as node sets endowed with the parent
function, rather than graphs with prescribed roots. As seen in
\cref{ex:transduction}, these two encodings can be translated to each
other by means of almost quantifier-free transductions, which render
them essentially equivalent.

\paragraph*{Normal forms.}  It will sometimes be useful to assume a
certain normal form of transductions.  We will need two similar, yet
slightly different normal forms: one for general transductions and one
for almost quantifier-free transductions. The proofs are standard, for
completeness, we give them in the appendix.

\begin{lemma}[$\star$]\label{lem:normal}
  Let $\interp I$ be a transduction. Then $\interp I$ is equivalent to
  a transduction of the form
  $$\interp L;\interp C;\interp F;\interp E;\interp X;\interp R,$$
  where
  \begin{itemize}
  \item $\interp L$ is a sequence of unary lifts;
  \item $\interp C$ is a sequence of copying operations;
  \item $\interp F$ is a sequence of function extension operations,
    one for each function on the output;
  \item $\interp E$ is a sequence of extension operations, one for
    each relation on the output;
  \item $\interp X$ is a single restriction operation; and
  \item $\interp R$ is a sequence of reduct operations.
  \end{itemize}
  Moreover, formulas parameterizing atomic operations in
  $\interp F;\interp E;\interp X$ use only relations and functions
  that appeared originally on input or were introduced by
  $\interp L;\interp C$.  In particular, none of these formulas uses
  any function or relation introduced by an atomic operation in
  $\interp F;\interp E$.
\end{lemma}

\begin{lemma}[$\star$]\label{lem:normal-qf}
  Every almost quantifier-free transduction is equivalent to an almost
  quantifier-free transduction that first applies a sequence of unary
  lifts and then applies a deterministic almost quantifier-free
  transduction.
\end{lemma}

\subsection{Treedepth and shrubdepth}\label{subsec:tdsd}
The \emph{treedepth} of a graph $G$ is the minimal depth of a rooted
forest $F$ with the same vertex set as $G$, such that for every edge
$uv$ of $G$, $u$ is an ancestor of $v$, or $v$ is an ancestor of $u$
in $F$.  A class $\CCC$ of graphs has \emph{bounded treedepth} if
there is a bound $d \in \N$ such that every graph in $\CCC$ has
treedepth at most $d$. Equivalently, $\Cc$ has bounded treedepth if
there is some number $k$ such that no graph in $\CCC$ contains a
simple path of length~$k$~\cite{Sparsity}.  The notion of treedepth
lifts to structures: a class $\CCC$ of structures has bounded
treedepth if the class of their Gaifman graphs has bounded treedepth.

\paragraph*{Shrubdepth.}  The following notion of \emph{shrubdepth}
has been proposed in~\cite{Ganian2012} as a dense analogue of
treedepth.  Originally, shrubdepth was defined using the notion of
\emph{tree-models}. We present an equivalent definition basing on the
notion of \emph{connection models}, introduced in~\cite{Ganian2012}
under the name of {\em{$m$-partite cographs}} of bounded depth.

\smallskip
A {\em connection model} with labels from $\Lambda$ is a rooted
labeled tree $T$ where each leaf $x$ is labeled by a label
$\lambda(x)\in \Lambda$, and each non-leaf node $v$ is labeled by a
(symmetric) binary relation $C(v)\subset \Lambda\times \Lambda$.  Such
a model defines a graph $G$ on the leaves of $T$, in which two
distinct leaves $x$ and $y$ are connected by an edge if and only if
$(\lambda(x),\lambda(y))\in C(v)$, where~$v$ is the least common
ancestor of $x$ and $y$.  We say that $T$ is a \emph{connection model}
of the resulting graph $G$.

\begin{example}
  Fix $n\in\N$, and let $G_n$ be the bi-complement of a matching of
  order $n$, i.e., the bipartite graph with nodes $a_1,\ldots,a_n$ and
  $b_1,\ldots,b_n$, such that $a_i$ is adjacent to~$b_j$ if and only
  if $i\neq j$.  A connection model for $G_n$ is shown below:
\begin{center}
	\includegraphics[height=19mm]{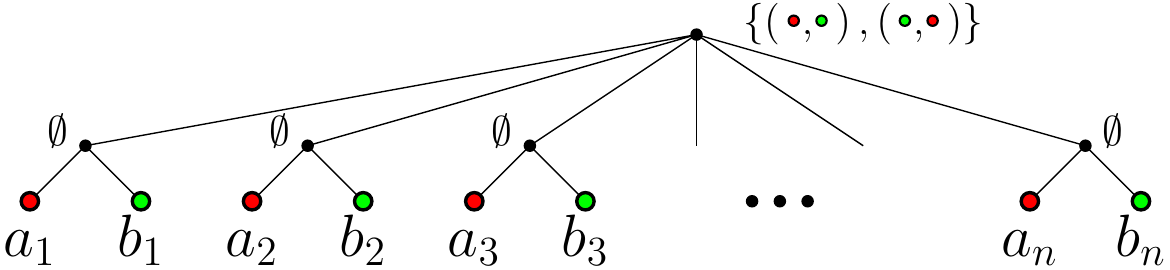}
\end{center}  
 \end{example}

 We can naturally extend the definition above to structures with unary
 functions by regarding each unary function by a binary relation
 selecting all $(\mathrm{argument},\mathrm{value})$ pairs.

\medskip
A class of graphs $\class{C}$ has {\em bounded shrubdepth} if there is
a number $h\in \N$ and a finite set of labels $\Lambda$ such that
every graph $G\in \CCC$ has a connection model of depth at most $h$
using labels from $\Lambda$.

%

\medskip Shrubdepth can be equivalently defined in terms of another
graph parameter, as follows.
%
%
Given a graph $G$ and a set of vertices $W\subset V(G)$, the graph
obtained by \emph{flipping the adjacency within $W$} is the graph $G'$
with vertices $V(G)$ and edge set which is the symmetric difference of
the edge set of $G$ and the edge set of the clique on $W$.
  
\smallskip The \emph{subset-complementation depth}, or
\emph{SC-depth}, of a graph is defined inductively as follows:
\begin{itemize}
\item a graph with one vertex has SC-depth $0$, and
\item a graph $G$ has SC-depth at most $d$, where $d\ge 1$, if there
  is a set of vertices $W\subset V(G)$ such that in the graph obtained
  from $G$ by flipping the adjacency within $W$ all connected
  components have SC-depth at most $d-1$.
\end{itemize}
  
\begin{example}
  A star has SC-depth at most $2$: flipping the adjacency within the
  set consisting of the vertices of degree $1$ yields a clique, which
  in turn has SC-depth at most~$1$.
\end{example}
  
The notion of SC-depth leads to a natural notion of decompositions.
An \emph{SC-decomposition} of a graph $G$ of SC-depth at most $d$ is a
rooted tree $T$ of depth $d$ with leaf set~$V(G)$, equipped with unary
predicates $W_0,\ldots,W_d$ on the leaves.  Each child $s$ of the root
in $T$ corresponds to a connected component $C_s$ of the graph $G'$
obtained from~$G$ by flipping the adjacency within $W_0$, such that
the subtree of $T$ rooted at $s$, together with the unary predicates
$W_1,\ldots,W_d$ restricted to $V(C_s)$, form an SC-decomposition of
$C_s$.

\medskip We will make use of the following properties, where the first
one follows from the definition of shrubdepth, and the remaining ones
follow from~\cite{Ganian2012}.

\begin{proposition}
  \label{prop:sd_properties}Let $\CCC$ be a class of graphs. Then:
  \begin{enumerate}
  \item\label{SD:1} If $\CCC$ has bounded shrubdepth then the class of
    all induced subgraphs of graphs from~$\CCC$ also has bounded
    shrubdepth.
  \item\label{SD:2} $\CCC$ has bounded shrubdepth if and only if for
    some $d\in \N$ all graphs in $\CCC$ have SC-depth at most $d$.
  \item\label{SD:3} If $\CCC$ has bounded treedepth then $\CCC$ has
    bounded shrubdepth.
  \item\label{SD:5} If $\CCC$ has bounded shrubdepth and $\interp I$
    is a transduction that outputs colored graphs, then
    $\interp I(\CCC)$ has bounded
    shrubdepth.
\end{enumerate}
\end{proposition}

It is well-known (see~\cite{Gurski2000}) that in the absence of large
bi-cliques (complete bipartite graphs) a graph of bounded cliquewidth
has in fact bounded treewidth.  The same holds also for shrubdepth and
treedepth. The lemma is proved by an easy induction on the depth of
the connection models.

\begin{lemma}[$\star$]\label{lem:bicliques}
  A class of graphs $\CCC$ has bounded treedepth if and only if graphs
  in $\CCC$ have bounded shrubdepth and exclude some fixed bi-clique
  as a subgraph.
\end{lemma}

\subsection{Bounded expansion}

A graph $H$ is a \emph{depth-$r$ minor} of a graph $G$ if $H$ can be
obtained from a subgraph of $G$ by contracting mutually disjoint
connected subgraphs of radius at most $r$. A class $\CCC$ of graphs
has \emph{bounded expansion} if there is a function
$f: \mathbb{N} \to \mathbb{N}$ such that
\(\frac{|E(H)|}{|V(H)|}\leq f(r)\)
for every $r \in \mathbb{N}$ and every depth-$r$ minor \(H\)
of a graph from $\CCC$.  Examples include the class of planar graphs,
or any class of graphs with bounded maximum degree.


\medskip We will use the following lemma.

\begin{lemma}\label{lem:lex}
  Let $\CCC$ be a class of (colored) graphs of bounded expansion and
  let $\interp C$ be a copy operation. Then $\interp C(\CCC)$ is a
  class of colored graphs of bounded expansion.
\end{lemma}
\begin{proof}
  Let $G\in \CCC$.  The Gaifman graph of $\interp C(G)$ is a subgraph
  of the so-called {\em{lexicographic product of $G$ with $K_2$}},
  i.e., it is constructed from the latter by replacing every vertex
  with two clones of it. It is known that if a class of graphs $\CCC$
  has bounded expansion, then the class of lexicographic products of
  graphs from $\CCC$ with any fixed graph $H$ also has bounded
  expansion; see e.g.,~\cite[Proposition 4.6]{Sparsity}.
\end{proof}

\smallskip The connection between treedepth and graph classes of
bounded expansion can be established via \emph{$p$-treedepth
  colorings}. For an integer $p$, a function $c: V(G) \to C$ is a
$p$-treedepth coloring if, for every $i\le p$ and set $X\subset V(G)$
with $|c(X)|=i$, the induced graph $G[X]$ has treedepth at most $i$. A
graph class $\CCC$ has \emph{low treedepth colorings} if for every
$p \in \mathbb{N}$ there is a number $N_p$ such that for every
$G \in \CCC$ there exists a $p$-treedepth coloring $c\from V(G)\to C$
with $|C|\le N_p$.

\begin{theorem}[\cite{POMNI}]\label{thm:beltc}
  A class of graphs $\CCC$ has bounded expansion if, and only if, it
  has low treedepth colorings.
\end{theorem}


%% file: sbe_covers_full.tex
\section{Main results}\label{sec:sbe}

In this section we introduce two notions which generalize the concept
of bounded expansion. Then we state the main results and outline the
proof.  First, we introduce classes of \emph{structurally bounded
  expansion}. This notion arises from closing bounded expansion graph
classes under transductions.

\begin{definition}
  A class $\CCC$ of graphs has {\em{structurally bounded expansion}}
  if there exists a class of graphs $\DDD$ of bounded expansion and a
  transduction $\interp I$ such that $\CCC\subseteq \interp I(\DDD)$.
\end{definition}

The second notion, \emph{low shrubdepth covers}, arises from the low
treedepth coloring characterisation of bounded expansion (see
\cref{thm:beltc}) by replacing treedepth by its dense counterpart,
shrubdepth. For convenience, we formally define this in terms of
\emph{covers}.


\begin{definition}\label{def:cover}
  A \emph{cover} of a graph $G$ is a family $\cal U_{G}$ of subsets of
  $V(G)$ such that $\bigcup \cal U_{G}= V(G)$.  A cover $\cal U_{G}$
  is a \emph{$p$-cover}, where $p \in \N$, if every set of at most $p$
  vertices is contained in some $U \in \cal U_{G}$.
  %
  If $\CCC$ is a class of graphs, then a \mbox{($p$-)}cover of $\CCC$
  is a family $\cal U=(\cal U_{G})_{G\in \CCC}$, where $\cal U_{G}$ is
  a ($p$-)cover of $G$.  The cover $\cal U$ is \emph{finite} if
  $\sup\set{|\cal U_{G}|\colon G\in \CCC}$ is finite.  Let
  $\CCC[\cal U]$ denote the class of graphs
  $\set{G[U]\colon G\in\CCC, U\in \cal U_{G}}$.  We say that the
  cover~$\cal U$ has \emph{bounded treedepth} (respectively, bounded
  shrubdepth) if the class $\CCC[\cal U]$ has bounded treedepth
  (respectively, shrubdepth).
\end{definition}

%

\begin{example}
  Let $\TTT$ be the class of trees and let $p\in\N$. We construct a
  finite $p$-cover~$\cal U$ of $\TTT$ which has bounded treedepth.
  Given a rooted tree $T$, let $\cal U_T=\set{U_0,\ldots,U_{p}}$,
  where~$U_i$ is the set of vertices of $T$ whose depth is not
  congruent to~$i$ modulo $p+1$. Note that $T[U_i]$ is a forest of
  height $p$, and that $\cal U_T$ is a $p$-cover of $T$.  Hence
  $\cal U=(\cal U_T)_{T\in\TTT}$ is a finite $p$-cover of $\TTT$ of
  bounded treedepth.
\end{example}

In analogy to low treedepth colorings, we can now characterize graph
classes of bounded expansion using covers.  We say that a class $\CCC$
of graphs has \emph{low treedepth covers} if for every $p\in \N$ there
is a finite $p$-cover $\cal U$ of $\CCC$ with bounded treedepth.  The
following lemma follows easily from \cref{thm:beltc}.

\begin{lemma}[$\star$]\label{lem:ltd-lsd-covers}
  A class of graphs has bounded expansion if, and only if, it has low
  treedepth covers.
\end{lemma}

We now define the second notion generalizing the concept of bounded
expansion. The idea is to use low shrubdepth covers instead of low
treedepth covers.

\begin{definition}\label{def:lsdc}
  A class $\CCC$ of graphs has \emph{low shrubdepth covers} if, and
  only if, for every $p\in \N$ there is a finite $p$-cover $\cal U$ of
  $\CCC$ with bounded shrubdepth.
\end{definition}

It is easily seen that \cref{lem:ltd-lsd-covers} together with
\cref{prop:sd_properties}(\ref{SD:3}) imply that every class of
bounded expansion has low shrubdepth covers.  Our main result is the
following theorem.

 \begin{theorem}\label{thm:main}
   A class of graphs has structurally bounded expansion if, and only
   if, it has low shrubdepth covers.
 \end{theorem}

 As a byproduct of our proof of \cref{thm:main} we obtain the
 following quantifier-elimination result, which we believe is of
 independent interest.

 \begin{theorem}\label{thm:qe-lsc}
   Let $\CCC$ be a class of colored graphs which has low shrubdepth
   covers. Then every transduction $\interp I$ is equivalent to some
   almost quantifier-free transduction $\interp J$ on~$\CCC$.
 \end{theorem}

%% file: outline.tex
We now outline the proof of \cref{thm:main} and \cref{thm:qe-lsc}.
Both theorems follow easily from \cref{pro:bi-def0} and
\cref{pro:qe-be0} stated below.  These are proved in subsequent
sections.

We start with the following lemma, which intuitively shows that covers
commute with almost quantifier-free transductions.

\begin{lemma}\label{lem:lsc0}
  If a class of graphs $\CCC$ has low shrubdepth covers and
  $\interp I$ is an almost quantifier-free transduction that outputs
  colored graphs, then $\interp I(\CCC)$ also has low shrubdepth
  covers.
\end{lemma}
\begin{proof}[sketch]
  The idea is that for any almost quantifier-free transduction
  $\interp I$
  there is a constant $c$ such any induced substructure of
  $\interp I(G)$ on $p$ elements depends only on an induced
  substructure of $G$ of size $p\cdot c$. In particular, a
  $(p\cdot c)$-cover of $G$ induces a $p$-cover of $\interp I(G)$.
  Moreover, as having bounded shrubdepth is preserved by
  transductions, a low shrubdepth cover of $\CCC$ induces a low
  shrubdepth cover of $\interp I(\CCC)$.  The details are presented in
  \cref{sec:transductions-and-covers}.
\end{proof}

\medskip The main novel ingredient in our proof of \cref{thm:main} and
\cref{thm:qe-lsc} is the following result, which intuitively states
that classes with low shrubdepth covers are bi-definable with classes
of bounded expansion, using almost quantifier-free transductions.

\begin{proposition}\label{pro:bi-def0}
  Suppose $\CCC$ is a class of graphs with low shrubdepth covers. Then
  there is a pair of transductions $\interp S$ and $\interp I$, where
  $\interp S$ is almost quantifier-free and $\interp I$ is
  deterministic almost quantifier-free, such that $\interp S(\CCC)$ is
  a class of colored graphs of bounded expansion and
  $\interp I(\interp S(G))=\{G\}$ for each $G\in \CCC$.
\end{proposition}
Clearly, \cref{pro:bi-def0} implies that $\CCC$ has structurally
bounded expansion, since it can be obtained as a result of
transduction $\interp I$ to a class $\interp S(\CCC)$ of bounded
expansion.  Thus, the right-to-left implication of \cref{thm:main} is
a corollary of the proposition.  The proof of \cref{pro:bi-def0} is
presented in \cref{sec:bi-def}. We sketch the rough idea below.

\begin{proof}[sketch]
  First, in \cref{lem:key} of \cref{sec:key}, we prove the special
  case where~$\CCC$ is a class of graphs of bounded shrubdepth, and
  for those we prove bi-definability with classes of trees of bounded
  depth. In particular, if $\DDD$ is a class of graphs of bounded
  shrubdepth, then there is a pair of almost quantifier-free
  transductions $\interp T,\interp I_0$ such that $\interp T(\DDD)$ is
  a class of colored trees of bounded depth and such that
  $\interp I_0(\interp T(H))=\set H$ for all $H\in \DDD$.
  \cref{lem:key} is the combinatorial core of this paper.

  To prove \cref{pro:bi-def0}, we lift \cref{lem:key} to the general
  case using covers, as follows. Let $\CCC$ be a class with low
  shrubdepth covers and let $\cal U$ be a $2$-cover of $\CCC$ of
  bounded shrubdepth, and let $N$ be such that $|\cal U_G|\le N$ for
  $G\in \CCC$.  We apply the bounded shrubdepth case to the class
  $\DDD=\CCC[\cal U]$, yielding almost quantifier-free transductions
  $\interp T$ and $\interp I_0$ as above. The transduction $\interp S$
  works as follows: given a graph $G\in\CCC$, introduce $N$ unary
  predicates marking the cover $\cal U_G$ of $G$, and for each
  $U\in \cal U_G$, apply $\interp T$ to the induced subgraph $G[U]$ of
  $G$, yielding a colored tree $\interp T(G[U])$.  Define
  $\interp S(G)$ as the union of the trees $\interp T(G[U])$, for
  $U\in \cal U_G$.  As $\cal U_G$ is a $2$-cover of $G$, $G$ is the
  union of the induced graphs $G[U]$ for $U\in \cal U_G$.  As each
  graph $G[U]$ can be recovered from the tree $\interp T(G[U])$ using
  the inverse transduction $\interp I_0$, it follows that $G$ can be
  recovered from the union $\interp S(G)$. This yields the inverse
  transduction $\interp I$ such that $\interp I(\interp S(G))=\set G$.
  As $\interp S$ is almost quantifier-free by construction, it follows
  from \cref{lem:lsc0} that $\interp S(\CCC)$ is a class with low
  shrubdepth covers. Moreover, each graph in $\interp S(\CCC)$ is a
  union of at most $N$ trees, so it does not contain $K_{N+1,N+1}$ as
  a subgraph. It follows from \cref{lem:bicliques} that the low
  shrubdepth cover of $\interp S(\CCC)$ is in fact a low treedepth
  cover. Hence, $\interp S(\CCC)$ has low treedepth covers, i.e., has
  bounded expansion.  
\end{proof}

\bigskip \cref{thm:qe-lsc}, and the remaining implication in
\cref{thm:main} are consequences of the following result.

\begin{proposition}\label{pro:qe-be0}
  Let $\CCC$ be a class of graphs of bounded expansion and let
  $\interp I$ be a transduction.  Then $\interp I$ is equivalent to an
  almost quantifier-free transduction $\interp J$ on $\CCC$.
\end{proposition}
We note that \cref{pro:qe-be0} is a strengthening of similar
statements provided by Dvo\v{r}\'ak et al.~\cite{DKT2} and of Grohe
and Kreutzer~\cite{Grohe2011}, and could be derived by a careful
analysis of their proofs.  In \cref{sec:qe} we provide a
self-contained proof, which we believe is simpler than the previous
proofs, and is sketched below.
 
\begin{proof}[sketch]
  We use the characterization of bounded expansion classes as those
  which have low treedepth covers.  We first prove \cref{pro:qe-be0}
  for forests of bounded depth. This can be handled by a direct
  (although slightly cumbersome) combinatorial argument, similarly as
  in~\cite{DKT2}. In Appendix~\ref{sec:automata} we present an
  argument using tree automata.
 
  The statement for classes of forests of bounded depth then easily
  lifts to classes of bounded treedepth. Here we use the fact that in
  a graph of bounded treedepth it is possible to encode a depth-first
  search forest of bounded depth, by using unary predicates marking
  the depth of each node in the spanning forest.
 
  We then lift the result from classes of bounded treedepth using
  covers. Specifically, suppose for simplicity that the transduction
  $\interp I$ is a single extension operation, parametrized by a
  formula $\psi$. We then proceed by induction on the structure of the
  formula $\psi$ and show that it can be replaced by a quantifier-free
  formula, at the cost of introducing unary functions defined by an almost
  quantifier-free transduction.

  In the inductive step, the only nontrivial case is the one of
  existential quantification, i.e., of formulas of the
  form $$\psi(\bar y)=\exists x.\phi(x,\bar y),$$ where
  $\phi(x,\bar y)$ may be assumed to be a quantifier-free formula
  involving unary functions, by inductive assumption. We consider a
  $p$-cover $\cal U$ of $\CCC$ where $p$ is a constant such that there
  are at most $p$ different terms occurring in $\phi(x,\bar y)$.
  Since $\CCC$ has bounded expansion, we may assume that the cover
  $\cal U$ has bounded treedepth, and that there is a constant
  $N\in \N$ such that $|\cal U_G|\le N$ for all $G\in \CCC$.  For a
  fixed graph $G\in\CCC$, the existentially quantified variable $x$
  must be in one of the sets $U\in \cal U_G$.  Therefore, the formula
  $\psi(\bar y)$ is equivalent to a disjunction of at most $N$
  formulas $\psi_{i}(\bar y)$, for $i=1,\ldots,N$, where each formula
  $\psi_{i}(\bar y)$ performs existential quantification restricted to
  the $i$th set in $\cal U_G$ (where $\cal U_G$ is ordered
  arbitrarily).  By the special case of the proposition proved for
  classes of bounded treedepth, $\psi_i(\bar y)$ is equivalent to a
  quantifier-free formula on $\CCC[\cal U]$ (the quantifier-free
  formula uses unary functions introduced by almost quantifier-free
  transductions). Reassuming, $\psi$ is equivalent on $G$ to a
  disjunction of quantifier-free formulas involving unary functions
  that are introduced by almost quantifier-free transductions. This
  deals with the inductive step. 
\end{proof}

We finally show how to conclude \cref{thm:main} and \cref{thm:qe-lsc}
from \cref{lem:lsc0}, \cref{pro:bi-def0} and \cref{pro:qe-be0}.

\begin{proof}[of \cref{thm:main}]As observed, the right-to-left
  implication of \cref{thm:main} follows from \cref{pro:bi-def0}.  We
  now show the left-to-right implication.

  Let $\CCC$ be a class of bounded expansion and let $\interp I$ be a
  transduction that outputs colored graphs.  We show that
  $\interp I(\CCC)$ has low shrubdepth covers.

  By \cref{lem:ltd-lsd-covers}, $\CCC$ has low treedepth covers.
  Applying \cref{pro:qe-be0} yields an almost quantifier-free
  transduction $\interp J$ such that
  $\interp I(\CCC) = \interp J(\CCC)$.  As $\CCC$ in particular has
  low shrubdepth covers (cf. \cref{prop:sd_properties}~(\ref{SD:3})),
  we may apply \cref{lem:lsc0} to $\interp J$ and $\CCC$ to deduce
  that $\interp J(\CCC)=\interp I(\CCC)$ has low shrubdepth covers.
\end{proof}
  
\begin{proof}[of \cref{thm:qe-lsc}]
  \cref{pro:bi-def0} allows to reduce the theorem to the case of
  classes of bounded expansion, as almost quantifier-free
  transductions are closed under composition. The case of bounded
  expansion classes is handled by \cref{pro:qe-be0}.
\end{proof}

It remains to provide the details of the proofs of \cref{lem:lsc0},
\cref{pro:bi-def0} and \cref{pro:qe-be0}.  This is done in
\cref{sec:transductions-and-covers}, \cref{sec:bi-def} and
\cref{sec:qe}, respectively.  After that, in \cref{sec:alg} we
conclude with a preliminary algorithmic result concerning the
model-checking problem for first-order logic on classes with
structurally bounded expansion.

%% file: definability_full.tex
\section{Proof of \cref{lem:lsc0}\quad
\normalfont\itshape\normalsize(almost quantifier-free transductions commute with covers)}\label{sec:transductions-and-covers}

In this section we prove \cref{lem:lsc0}, which we restate for
convenience.

\setcounter{tmp}{\thetheorem}
\setcounterref{theorem}{lem:lsc0}
\addtocounter{theorem}{-1}
\begin{lemma}\label{lem:lsc}
  If a class of graphs $\CCC$ has low shrubdepth covers and
  $\interp I$ is an almost quantifier-free transduction that outputs
  colored graphs, then $\interp I(\CCC)$ also has low shrubdepth
  covers.
\end{lemma}
 \setcounter{theorem}{\thetmp}

 We start with formulating the following lemma which states that
 almost quantifier-free transductions are, in a certain sense, local.

\begin{lemma}\label{lem:dep}
  For every deterministic almost quantifier-free transduction
  $\interp I$ there is a constant $c\in \N$ such that the following
  holds.  For every structure $\str A$ and every element~$v$ of
  $\interp I(\str A)$ there is a set $S_v\subset V(\str A)$ of size at
  most $c$ such that for any sets $U,W$ with
  $W\subset V(\interp I(\str A))$ and $U\subset V(\str A)$, if
  $U\supseteq\bigcup_{v\in W}S_v$, then
$$\interp I(\str A)[W]=\interp I(\str A[U])[W].$$
\end{lemma}

\newcommand{\depend}{\mathsf{cl}} 

In order to prove the lemma, we define the following notions of
\emph{dependency} and \emph{support}.

\begin{definition}
  Suppose that $\tau(v)=(f_p\circ\cdots \circ f_1)(v)$ is a term.  For
  a structure~$\str A$ carrying partial functions $f_1,\ldots,f_p$, we
  say that an element $v\in V(\str A)$ {\em{$\tau$-depends}} with
  respect to $\tau$ on itself and all elements of the form
  $(f_p\circ\cdots \circ f_i)(v)$ for $i\in [p]$, whenever defined.
  For a quantifier-free formula $\varphi(x_1,\ldots,x_k)$, an element
  $v\in V(\str A)$ {\em{$\varphi$-depends}} on all elements on which
  $v$ $\tau$-depends, for any term $\tau$ appearing in~$\varphi$.  For
  an element $v$, the set of elements on which $v$ $\varphi$-depends
  in $\str A$ will be denoted by $\depend^{\str A}_\varphi(v)$; note
  that the size of this set is always bounded by a constant depending
  only on $\varphi$.  Observe also that given elements
  $v_1,\ldots,v_k$, to check whether $\varphi(v_1,\ldots,v_k)$ holds
  in $\str A$ it suffices to check whether it holds in the
  substructure of~$\str A$ induced by all elements on which
  $v_1,\ldots,v_k$ $\varphi$-depend.
\end{definition}

With the auxiliary notion of dependency defined we can come to the
definition of \emph{support. }

\begin{definition}
  Suppose $\interp I$ is a deterministic almost quantifier-free
  transduction, and let $\str A$ be an input structure.  For an
  element $v\in V(\interp I(\str A))$ and a subset
  $S\subseteq V(\str A)$, we now define what it means that $v$ {\em{is
      $\interp I$-supported}} by $S$.  We first define this for atomic
  operations (note that unary lifts are excluded since $\interp I$ is
  assumed to be deterministic):
  \begin{itemize}
  \item If $\interp I$ is a reduct operation or a copy operation, then
    $v$ is $\interp I$-supported by $S$ if and only if $v\in S$.

  \item If $\interp I$ is a restriction or an extension operation, say
    parameterized by a formula $\varphi$, then $v$ is
    $\interp I$-supported by $S$ if and only if
    $\depend^{\str A}_\varphi(v)\subseteq S$.


  \item Suppose $\interp I$ is a function extension operation, say
    introducing a partial function~$f$ using a binary formula
    $\varphi(x,y)$.  Then $v$ is $\interp I$-supported by $S$ if and
    only if $\depend^{\str A}_\varphi(v)\subseteq S$ and the following
    holds:
    \begin{itemize}
    \item if there exists exactly one $w\in V(\str A)$ for which
      $\varphi(v,w)$ holds, then
      $\depend^{\str A}_\varphi(w)\subseteq S$.
    \item if there are at least two elements $w\in V(\str A)$ for
      which $\varphi(v,w)$ holds, then
      $\depend^{\str A}_\varphi(w)\subseteq S$ for at least two
      distinct such elements $w$.
    \end{itemize}
  \end{itemize}

  Finally, for non-atomic deterministic almost quantifier-free
  transductions the notion of $\interp I$-supporting is defined by
  induction on the structure of the transduction.  Suppose~$\interp I$
  is the composition $\interp I_1; \interp I_2$ of two transductions.
  Then $v\in V(\interp I(\str A))$ is $\interp I$-supported by
  $S\subseteq V(\str A)$ if there exists a subset
  $T\subseteq V(\interp I_1(\str A))$ and, for each $w\in T$, a subset
  $S_w\subseteq S$ such that $v$ is $\interp I_2$-supported by $T$ and
  each $w\in T$ is $\interp I_1$-supported by $S_w$.
\end{definition}

The notion of supporting is trivially closed under taking supersets:
if $v$ is $\interp I$-supported by $S$, then $v$ is also
$\interp I$-supported by any superset of $S$.

\begin{proof}[of \cref{lem:dep}]
  By induction on the definition of an almost quantifier-free
  transduction $\interp I$ it is easy to see that for every
  $v\in V(\interp I(\str A))$ there is a set $S_v\subset V(\str A)$
  such that $v$ is $\interp I$-supported by $S_v$ and $|S_v|$ is
  bounded by a constant, possibly depending on $\interp I$.
	
  By induction we also observe that if $W\subset V(\interp I(\str A))$
  and $U\subset V(\str A)$ are such that every $v\in W$ is
  $\interp I$-supported by $U$ then
  $$\interp I(\str A)[W]=\interp I(\str A[U])[W].$$		
  This proves the lemma.
\end{proof}

\medskip We can now prove \cref{lem:lsc}.
\begin{proof}[of \cref{lem:lsc}]
  Let $\CCC$ be a class with low shrubdepth covers and let $\interp I$
  be an almost quantifier-free transduction that outputs colored
  graphs. We show that $\interp I(\CCC)$ has low shrubdepth covers.
  By normalizing $\interp I$ as described in \cref{lem:normal-qf}, we
  may assume that $\interp I$ is of the form $\interp L;\interp J$,
  where $\interp L$ is a sequence of unary lifts and $\interp J$ is
  deterministic almost quantifier-free.  As~$\CCC$ has low shrubdepth
  covers, the class $\DDD=\interp L(\CCC)$ also has low shrubdepth
  covers (this is implied by \cref{prop:sd_properties}(\ref{SD:5})).
  Moreover, $\interp I(\CCC)=\interp J(\DDD)$. Therefore, it suffices
  to focus on the deterministic almost quantifier-free
  transduction~$\interp J$ applied to the class~$\DDD$.  Note that
  $\DDD$ is a class of colored graphs, i.e., graphs with unary
  predicates on their vertices.

  Let $c$ be the constant provided by \cref{lem:dep} for the
  transduction $\interp J$.  We need to find, for every $p\in \N$, a
  finite $p$-cover of $\interp J(\DDD)$ of bounded shrubdepth, so let
  us fix~$p$.  Let~$\cal U$ be a finite $(c\cdot p)$-cover of $\DDD$
  of bounded shrubdepth.  For a graph $G\in \DDD$ and $U\in \cal U_G$,
  let $W_U\subseteq V(\interp J(G))$ be the set of those elements $v$
  of $\interp J(G)$ such that $S_v\subset U$, where~$S_v$ is as
  obtained from \cref{lem:dep} applied to the deterministic almost
  quantifier-free transduction $\interp J$.

  Define a cover $\cal W=(\cal W_{\interp J(G)})_{G\in \DDD}$ of
  $\interp J(\DDD)$ by letting
  $$\cal W_{\interp J(G)} = \{ W_U\colon U\in \cal U_G\}
  \qquad \textrm{for every graph $G\in \DDD$.}$$
  Clearly $|\cal W_{\interp J(G)}|\leq |\cal U_G|$, so $\cal W$ is
  finite as well. We need to verify that $\cal W$ is a $p$-cover and
  that it has bounded shrubdepth.

  To see that $\cal W$ is a $p$-cover, take any $p$ elements
  $w_1,\ldots,w_p$ of $\interp J(G)$.  Let
  $S=\bigcup_{i=1}^p S_{w_i}$. Then $|S|\leq c\cdot p$, hence there
  exists $U\in \cal U_G$ with $S\subseteq U$.  We conclude that
  $\set{w_1,\ldots,w_p}\subset W_U \in \cal W_G$.

  To see that $\cal W$ is a bounded shrubdepth cover, observe that by
  assumption $\DDD[\cal U]$ has bounded shrubdepth, hence by
  \cref{prop:sd_properties}(\ref{SD:5}) we find that
  $\interp J(\DDD[\cal U])$ also has bounded shrubdepth.  By
  \cref{lem:dep}, for each $G\in \DDD$ and
  $W_U\in \cal W_{\interp J(G)}$, the induced substructure
  $\interp J(G)[W_U]$ is equal to $\interp J(G[U])[W_U]$. Now it
  suffices to note that $\interp J(G[U])\in \interp J(\DDD[\cal U])$,
  hence $\interp J(G)[W_U]$ belongs to the hereditary closure of
  $\interp J(\DDD[\cal U])$, which also has bounded shrubdepth by
  \cref{prop:sd_properties}(\ref{SD:1}).
\end{proof}

\section{Proof of \cref{pro:bi-def0}\quad \normalfont\itshape\normalsize(bi-definability of classes with low shrubdepth covers and classes of bounded expansion)}\label{sec:bi-def}

In this section we prove \cref{pro:bi-def0}, which we repeat for
convenience.

\setcounter{tmp}{\thetheorem}
\setcounterref{theorem}{pro:bi-def0}
\addtocounter{theorem}{-1}
 \begin{proposition}\label{pro:bi-def}
   Suppose $\CCC$ is a class of graphs with low shrubdepth
   covers. Then there is a pair of transductions $\interp S$ and
   $\interp I$, where $\interp S$ is almost quantifier-free and
   $\interp I$ is deterministic almost quantifier-free, such that
   $\interp S(\CCC)$ is a class of colored graphs of bounded expansion
   and $\interp I(\interp S(G))=\{G\}$ for each $G\in \CCC$.
 \end{proposition}
 \setcounter{theorem}{\thetmp}

 Clearly, \cref{pro:bi-def} implies that $\CCC$ has structurally
 bounded expansion, since it can be obtained as a result of
 transduction $\interp I$ to a class $\interp S(\CCC)$ of bounded
 expansion.  Thus, the right-to-left implication of \cref{thm:main} is
 a corollary of the proposition.

 \medskip The idea of the proof of \cref{pro:bi-def} is as follows.
 We first prove in \cref{lem:connectome} of \cref{sec:def-ccs} that
 connected components in graphs of bounded shrubdepth are definable by
 almost quantifier-free transductions.  We use \cref{lem:connectome}
 to first prove \cref{pro:bi-def} for the special case where $\CCC$ is
 a class of graphs of bounded shrubdepth, and for those we prove
 bi-definability with classes of trees of bounded depth. This is done
 in \cref{lem:key} of \cref{sec:key}.  Then, we conclude the general
 case in \cref{sec:key-general}, by lifting \cref{lem:key} using
 covers.


%
%

\subsection{Defining connected components in graphs of bounded shrubdepth}\label{sec:def-ccs}
The following lemma is the combinatorial core of our proof of \cref{pro:bi-def}.

\begin{lemma}\label{lem:connectome}
  Let $\CCC$ be a class of graphs of bounded shrubdepth.  There is an
  almost quantifier-free transduction $\interp F$ such that for a
  given $G\in \CCC$, every output of $\interp F$ on $G$ is equal to
  $G$ enriched by a function $g\from V(G)\to V(G)$ such that
  $g(v)=g(w)$ if and only if $v$ and $w$ are in the same connected
  component of $G$.
\end{lemma}
The rest of \cref{sec:def-ccs} is devoted to the proof of \cref{lem:connectome}.
\medskip

\paragraph{Guidance systems.}
We first introduce the notions of {\em{guidance systems}} and of
functions {\em{guided}} or {\em{guidable}} by them.  This is a
combinatorial abstraction for functions computable by almost
quantifier-free transductions.

Let $G$ be a graph. A {\em{guidance system}} in $G$ is any family
$\Uu$ of subsets of the vertex set of $G$.  The {\em{size}} of a
guidance system $\Uu$ is the cardinality of the family $\Uu$.  We say
that a partial function $f\colon V(G)\partto V(G)$ is {\em{guided}} by
the guidance system $\Uu$ if for every $x\in V(G)$ for which $f(x)$ is
defined and different than $x$, there is some $U\in\Uu$ such
that~$f(x)$ is the unique neighbor of $v$ in $U$.  Finally, a partial
function $f\from V(G)\partto V(G)$ is \emph{$\ell$-guidable}, where
$\ell\in\N$, if there is a guidance system $\Uu$ of size at most
$\ell$ in $G$ that such that $f$ is guided by $\Uu$.

Observe that an $\ell$-guidable partial function maps each vertex $v$
from its domain to a vertex in the same connected component as $v$.
The following lemmas will be useful for operating on guidable
functions.

\begin{lemma}[$\star$]\label{lem:glue-directed}
  Let $G$ be a graph and suppose $g\from V(G)\partto V(G)$ is a
  partial function such that the restriction $g|_C$ of $g$ to each
  connected component $C$ of $G$ is $\ell$-guidable.  Then $g$ is
  $\ell$-guidable.
\end{lemma}

\begin{lemma}[$\star$]\label{lem:glue-directable}
  Let $G$ be a graph and let $g_1,\ldots,g_s\from V(G)\partto V(G)$ be
  partial functions, where $g_i$ is $\ell$-guidable for each
  $i\in [s]$.  If $g\from V(G)\partto V(G)$ is a partial function such
  that for every $x\in V(G)$ there is some $i\in[s]$ such that
  $g(x)=g_i(x)$, then $g$ is $(\ell\cdot s)$-guidable.
\end{lemma}

Finally, guidable functions can be computed using almost
quantifier-free transductions.

\begin{lemma}[$\star$]\label{lem:local}
  Let $\CCC$ be a class of graphs and let $\ell\in\N$ be
  fixed. Suppose that each $G\in \CCC$ is equipped with an
  $\ell$-guidable function $f_G\from V(G)\partto V(G)$.  Then there
  exists an almost quantifier-free transduction which given
  $G\in \CCC$ has exactly one output: the graph $G$ enriched
  with~$f_G$.
\end{lemma}

We will use the following fact stating that graphs of bounded shrubdepth do not
admit long induced paths.

%
%

\begin{lemma}[\cite{Ganian2013}]\label{lem:ind-paths}
  For every class $\CCC$ of graphs of bounded shrubdepth there exists
  a constant $r\in \N$ such that no graph from $\CCC$ contains a path
  on more than $r$ vertices as an induced subgraph.  Consequently, for
  every graph $G\in \CCC$ every connected component of $G$ has
  diameter at most $r$.
\end{lemma}

\paragraph{Spanning forests.}
For a graph $G$ and a function $g\from V(G)\to V(G)$, we 
say that $g$ \emph{defines a spanning forest of depth $r$} on $G$
if $g$ is guarded by $G$ and 
the $r$-fold composition $g^r\from V(G)\to V(G)$ is constant when restricted to each connected component of~$G$.
In particular,  two vertices $u,v\in V(G)$ are in the same connected component of $G$
if and only if $g^r(u)=g^r(v)$.

The following lemma states that guidance systems can define 
shallow spanning forests in graph classes of bounded shrubdepth.

\begin{lemma}\label{lem:bfs}
  For every class $\CCC$ of graphs of bounded shrubdepth there exist
   constants $q,r\in \N$ such that for every $G\in \CCC$ there is 
  a function $f_G\from V(G)\to V(G)$ which is $q$-guidable as a partial
  function on $G$ and defines a spanning forest of depth $r$ on $G$.	
\end{lemma}

We first show how \cref{lem:connectome} follows from \cref{lem:bfs}.

\begin{proof}[of \cref{lem:connectome}]
By \cref{lem:local}, there is an almost
  quantifier-free transduction~$\interp I$ which, given a graph $G\in\CCC$ on input, constructs the
  function $f_G$ obtained from \cref{lem:bfs}. Now let $g=f_G^r$ be the
  $r$-fold composition of $f$. Clearly, $g$ can be computed by an
  almost quantifier-free transduction using a single function
  extension operation, making use of the function~$f_G$ constructed by
  $\interp I$.  As $g$ is constant on every connected component of $G$, \cref{lem:connectome} follows.
\end{proof}

It remains to prove \cref{lem:bfs}. 

\paragraph{Constructing guidable choice functions.}
\cref{lem:bfs} will follow easily from
the fact that connected components of graphs of bounded shrubdepth have bounded diameter by \cref{lem:ind-paths}, 
and from the following lemma, essentially stating that every total binary relation whose graph has  bounded shrubdepth
contains a guidable choice function.

\begin{lemma}\label{lem:go-right}
  For every class $\CCC$ of graphs of bounded shrubdepth there exists
  a constant $p\in \N$ such that the following holds.  Suppose
  $G\in \CCC$ and $A$ and $B$ are two disjoint subsets of vertices of
  $G$ such that every vertex of $A$ has a neighbor in $B$.  Then there
  is a function $f\from A\to B$ which is $p$-guidable as a partial
  function on $G$.
\end{lemma}

We found two conceptually different proofs of this result. We believe
that both proofs describe complementary viewpoints on the problem, so
we present both of them.  To keep the presentation concise, in the
main body of the paper we give only one proof, using the
characterization of classes of bounded shrubdepth using connection
models, and their close connection to \emph{bi-cographs}.
We present the second proof in Appendix~\ref{app:greedy}, which provides an explicit
greedy procedure leading to the construction of $f$.
 

\medskip
We first prove a special case of \cref{lem:go-right} for graphs which
have a connection model using two different labels $\alpha$ and
$\beta$, where one part of $G$ has label $\alpha$ and the other part
has label $\beta$. Such graphs are called {\em{bi-cographs}}
(cf.~\cite{Giakoumakis1997}).

\begin{lemma}\label{lem:go-right-bico}
  Let $G$ be a bi-cograph with parts $A,B$ and with a connection model
  of height~$h$ where vertices in $A$ have label $\alpha$ and vertices
  in $B$ have label $\beta$.  Suppose further that every vertex in $A$
  has a neighbor in $B$. Then there is a function $f\from A\to B$
  which is $h$-guidable as a partial function on $G$.
\end{lemma}
\begin{proof}
  By \cref{lem:glue-directed}, it is enough to consider the case when
  $G$ is connected. Let~$T$ be the assumed connection model of height
  $h$.

  We prove that there is an $h$-guidable function $f\from A\to B$.
  The proof proceeds by induction on $h$.  The base case, when $h=1$
  is trivial, because then every vertex of $A$ is adjacent to every
  vertex of $B$, so picking any $w\in B$ the function $f\from A\to B$
  which maps every $v\in A$ to $w$ is guided by the guidance system
  consisting only of $\set{w}$.
	
  In the inductive step, assume that $h\ge 2$ and the statement holds
  for height $h-1$.  Since~$G$ is connected, either the label $C(r)$
  of the root $r$ contains the pair $(\alpha,\beta)$, or~$r$ has only
  one child $v$. In the latter case, the subtree of $T$ rooted at $v$
  is a connection model of $G$ of height $h-1$, so the conclusion
  holds by inductive assumption.  Hence, we assume that
  $(\alpha,\beta)\in C(r)$.
	
  Let $\cal S$ be the set of bipartite induced subgraphs $H$ of $G$
  such that $H$ is defined by the connection model rooted at some
  child of $r$ in $T$.  As $(\alpha,\beta)\in C(r)$, it follows that
  if $H_1,H_2\in\cal S$ are two distinct graphs, then every vertex
  with label $\alpha$ in $H_1$ is connected to every vertex with label
  $\beta$ in $H_2$.  We consider two cases, depending on whether
  $\cal S$ contains more than one graph $H$ containing a vertex with
  label $\beta$, or not.

  In the first case, there are at least two graphs $H_1,H_2\in\cal S$
  such that $H_1$ and $H_2$ both contain a vertex with label
  $\beta$. Pick $w_1\in V(H_1)$ and $w_2\in V(H_2)$, both with label
  $\beta$.  Then every vertex in $A$ is adjacent either to $w_1$ or to
  $w_2$. Let $f\from A\to B$ be a function which maps a vertex
  $v\in A$ to $w_1$ if $v$ is adjacent to $w_1$, and to $w_2$
  otherwise. Then $f$ is guided by the guidance system consisting of
  $\set{w_1}$ and $\set{w_2}$.

  In the second case, there is only one graph $H\in\cal S$ which
  contains a vertex with label $\beta$.  Pick an arbitrary vertex $w$
  with label $\beta$ in $H$.  Notice that every vertex in $V(G)-V(H)$
  is adjacent to $w$.  The graph $H$ has a connection model of height
  $h-1$, so by inductive assumption, there is a guidance system
  $\cal U\subset\pow{V(H)}$ of size at most $h-1$ and a function
  $f_0\from V(H)\cap A\to V(H)\cap B$ which is guided by $\cal U$.
  Then the function $f\from A\to B$ which extends $f_0$ by mapping
  every vertex in $V(G)-V(H)$ to $w$ is guided by
  $\cal U\cup\set{\set{w}}$.  In either case, we have constructed a
  $h$-guidable function $f\from A\to B$, as required.
\end{proof}

We now prove \cref{lem:go-right} in the general case.

\begin{proof}[of \cref{lem:go-right}]\label{pf:}
  Let $\CCC$ be a class of graphs of bounded shrubdepth. Hence, there
  is a finite set of labels $\Lambda$ and a number $h\in\N$ such that
  every graph $G\in\CCC$ has a connection model of height $h$ using
  labels from $\Lambda$. For $\alpha\in\Lambda$, let $V_\alpha$ denote
  the set of vertices of $G$ which are labeled $\alpha$.
	
  Define a function $\mu\from A\to \Lambda^2$ as follows: for every
  vertex $v$ define $\mu(v)$ as $(\alpha,\beta)$, where $\alpha$ is
  the label of $v$, and $\beta\in\Lambda$ is an arbitrary label such
  that $v$ has a neighbor in~$B$ with label $\beta$.

  For every pair of labels $\alpha,\beta$, consider the bipartite
  graph $G_{\alpha\beta}$ which is the subgraph of $G$ consisting of
  $\mu^{-1}((\alpha,\beta))$ on one side and $B\cap V_\beta$ on the
  other side, and all edges between these sets; note that they are
  disjoint, as one is contained in $A$ and second in~$B$. Observe that
  $G_{\alpha\beta}$ is a bi-cograph with a connection model of height
  $h$, such that every vertex in $V(G_{\alpha\beta})\cap A$ has a
  neighbor in $V(G_{\alpha\beta})\cap B$.  By \cref{lem:go-right-bico}
  there is a function
  $f_{\alpha\beta}\from \mu^{-1}((\alpha,\beta)) \to B\cap V_\beta$
  which is $h$-guidable in $G_{\alpha\beta}$. Observe that
  $f_{\alpha\beta}$ is also $h$-guidable when treated as a partial
  function on $G$; it suffices to take the same guidance system, but
  with all its sets restricted to~$B$.

  Finally, define the function $f\from A\to B$ so that if $v\in A$ and
  $\mu(v)=(\alpha,\beta)$, then $f(v)=f_{\alpha\beta}(v)$.  By
  \cref{lem:glue-directable}, the function $f$ is
  $(h\cdot |\Lambda|^2)$-guidable. This concludes the proof of
  \cref{lem:go-right}.
\end{proof}

\paragraph{Constructing guidable spanning forests.}
We are ready to complete the proof of \cref{lem:bfs} stating that shallow 
spanning forests on classes of bounded shrubdepth are definable by guidance systems.

\begin{proof}[of \cref{lem:bfs}]
	Let $\CCC$ be a class of graphs of bounded shrubdepth, and 
  let $r$ and $p$ be constants provided by \cref{lem:ind-paths} and
  \cref{lem:go-right}, respectively, for the class~$\CCC$.  
  Let
  $R_0\subset V(G)$ be a set of vertices which contains exactly one
  vertex in each connected component $C$ of $G$.
By \cref{lem:ind-paths}, we may assume that 
every vertex in $G$ is at distance at most $r$ from a unique vertex in $R_0$.  
    For
  $i=1,\ldots,r$, let $R_i$ be the set of vertices of $G$ whose
  distance to some vertex in $R_0$ is equal to~$i$.  Then the sets $R_0,R_1,\ldots,R_r$ form a partition
  of the vertex set of $G$.  Furthermore, observe that for
  $i=1,\ldots,r$, every vertex of $R_i$ has a neighbor in $R_{i-1}$.

  Fix a number $i\in\set{1,\ldots,r}$. Apply \cref{lem:go-right} to
  $R_i$ as $A$ and $R_{i-1}$ as $B$.  This yields a function
  $f_i\from R_i\to R_{i-1}$ which is $p$-guidable in
  $G[R_i\cup R_{i-1}]$.  In particular, $f_i$ is also a $p$-guidable
  partial function $f_i\from V(G)\partto V(G)$.
  Let $f_0$ be a partial function from $V(G)$ to $V(G)$ that fixes
  every vertex of $R_0$ and is undefined otherwise.  Then $f_0$ is
  guided by the guidance system $\{R_0\}$, hence it is $1$-guidable in
  $G$.

  Consider now the function $f_G\from V(G)\to V(G)$ such that for
  $u\in V(G)$, $f_G(u)=f_i(u)$ if $f_i(u)$ is defined for some
  $i\in\set{0,\ldots,r}$.  By the first item of
  \cref{lem:glue-directable} we find that~$f_G$ is
  $p(r+1)$-guidable. By construction, $f_G$ is guarded, and
 $f_G^r$ maps every vertex $v\in V(G)$
  to the unique vertex in $R_0$ which lies in the connected component of $v$.
This proves that $f_G$ defines a spanning forest of depth $r$ on $G$.
\end{proof} 

This completes the proof of \cref{lem:connectome}.

\subsection{\cref{pro:bi-def} for classes of bounded shrubdepth}\label{sec:key}

In this section, we prove \cref{pro:bi-def} in the special case when
$\CCC$ is a class of graphs of bounded shrubdepth:
\begin{lemma}\label{lem:key}
  Let $\cal B$ be a class of graphs of bounded shrubdepth.  Then there
  is a class $\cal T$ of colored trees of bounded height and a pair of
  transductions $\interp T$ and $\interp B$ such that $\interp T$ is
  almost quantifier-free, $\interp B$ is deterministic almost
  quantifier-free, $\interp T(\cal B)\subseteq \cal T$,
  $\interp B(\cal T)\subseteq \cal B$, and
  \[
  \interp B(\interp T(G))=\{G\}\ \ \textrm{for all}\ \ G\in {\cal B}
  \quad \textrm{and}\quad \interp T( \interp B(t))\ni t\ \ \textrm{for
    all}\ \ t\in \cal T.
  \]
  Moreover, for any $G\in \cal B$, every $t\in \interp T(G)$ is an
  SC-decomposition of $G$.
\end{lemma}

We remark that in \cref{lem:key}, every output of the transduction
$\interp T$ is an SC-decomposition of the input graph of bounded
depth, whereas the transduction $\interp B$ recovers the graph from
its SC-decomposition.

In other words, the lemma allows to construct the SC-decomposition of
a graph from a class of graphs of bounded shrubdepth using an almost
quantifier-free transduction.
%
This argument is the combinatorial cornerstone of our approach.
Conceptually, it shows that bounded-height decompositions of graphs
from classes of bounded shrubdepth can be defined in a very weak
logic, as essentially the whole information about the decomposition
can be pushed to unary predicates on vertices (added using unary
lifts), and from this information the decomposition can be
reconstructed using only deterministic almost quantifier-free
formulas.

\medskip


We need one more auxiliary lemma which allows to apply a transduction
in parallel to a disjoint union of structures.  Suppose $\cal K$ is a
set of structures over the same signature.  The \emph{bundling} of
$\cal K$ is a structure obtained by taking the disjoint union
$\bigcup \cal K$ of the structures in~$\cal K$, extended with a set
$X$ disjoint from $V(\bigcup\cal K)$ and a function
$f\from V(\bigcup\cal K)\to X$ such that $f(x)=f(y)$ if and only if
$x,y$ belong to the same structure in $\cal K$.  We denote such a
bundling by $\bigcup \cal K^X$.  We now prove that an almost
quantifier-free transduction working on each structure separately can
be lifted to their bundling.

\begin{lemma}[$\star$]\label{lem:parallel}
  Let $\interp I$ be an almost quantifier-free transduction. Then
  there is an almost quantifier-free transduction ${\interp I}^\star$
  such that if the input to $\interp I^\star$ is the bundling
  $\bigcup \cal K^X$ of $\cal K$, then
  $\interp I^\star(\bigcup \cal K^X)$ is the set containing the
  bundling of every set formed by taking one member from
  $\interp I(\str K)$ for each $\str K\in \cal K$.
\end{lemma}

We can now give a proof of \cref{lem:key}.

\begin{proof}[of \cref{lem:key}]
  Let $\cal B_d$ be the class of graphs of SC-depth at most $d$.  We
  prove the statement for $\cal B=\cal B_d$, yielding appropriate
  transductions $\interp B_d$ and $\interp T_d$.  Observe that this
  implies the general case: if $\cal B$ is any class of graphs of
  bounded shrubdepth, then by \cref{prop:sd_properties}(\ref{SD:2})
  there is a number $d$ such that every graph from $\cal B$ has
  SC-depth at most $d$, hence we may set $\interp B=\interp B_d$,
  $\interp T=\interp T_d$, and $\cal T=\interp T(\cal B)$.

  The proof is by induction on $d$.  The base case, when $d=0$, is
  trivial.  In general, every output of $\interp T_d$ will be an
  SC-decomposition of the input graph of depth $d$.  That is, it is a
  tree of height $d$, here encoded as a structure by providing its
  parent function.  The leaves of this tree are exactly the original
  vertices of the input graph $G$.  They are colored with $d$ unary
  predicates $W_0,W_1,\ldots,W_{d-1}$, corresponding to flip sets used
  on consecutive levels of the SC-decomposition.

  Now, given an almost quantifier-free transduction $\interp T_d$ we
  construct an almost quantifier-free transduction $\interp T_{d+1}$.
  The transduction $\interp T_{d+1}$, given a graph $G$,
  nondeterministically computes a rooted tree $t_G$ as above in the
  following steps.  Implementing each of them using an almost
  quantifier-free transduction is straightforward, and to keep the
  description concise, we leave the implementation details to the
  reader.
  \begin{itemize}
  \item Since $G\in \cal B_{d+1}$, there is a vertex subset
    $W\subseteq V(G)$ such that in the graph $G'$ obtained from $G$ by
    flipping the adjacency within $W$ every connected component
    belongs to $\cal B_d$.  Using a unary lift, introduce a unary
    predicate $W_{0}$ selecting the set $W$ and compute $G'$ by
    flipping the adjacency within $W_{0}$.
  
  \item Let $g\from V(G')\to V(G')$ be the function given by
    \cref{lem:connectome}, applied to the graph~$G'$. Note that $g$
    can be constructed using an almost quantifier-free transduction.
    Using copying and restriction, create a copy $X$ of the image of
    $g$. By composing~$g$ with the function that maps each element of
    the image of~$g$ to its copy (easily constructible using function
    extension), we construct a function $g'\from V(G')\to X$ such that
    $g'(v)=g'(w)$ if and only if $v$ and $w$ are in the same connected
    component of $G'$.  Hence, $g'\from V(G')\to X$ defines a bundling
    of the set of connected components of~$G'$.
  
  \item Apply \cref{lem:parallel} to the transduction $\interp T_d$
    yielding a transduction $\interp T_d^\star$.  Our transduction
    $\interp T_{d+1}$ now applies $\interp T_d^\star$ to the bundling
    given by $g'$, resulting in a bundling of the family of colored
    trees $t_C$, for $C$ ranging over the connected components
    of~$G'$.

  \item Using extension, mark the roots of the trees $t_C$ with a new
    unary predicate; for $C$ ranging over the connected components of
    $G'$ these are exactly elements that do not have a parent. Create
    new edges which join each such a root $r$ with $g'(r)$. In effect,
    for every connected component $C$ of $G'$, all the roots of the
    trees $t_C$ are appended to a new root $r_C$.  At the end clear
    all unnecessary relations from the structure.  Note that the
    obtained tree~$t_G$ retains all unary predicates $W_1,\ldots,W_d$
    that were introduced by the application of the transduction
    $\interp T_d^\star$ to $G'$, as well as the predicate $W_{0}$
    introduced at the very beginning. All these predicates select
    subsets of leaves of $t_G$.
  \end{itemize}

  This concludes the description of the almost quantifier-free
  transduction $\interp T_{d+1}$.  The transduction $\interp B_{d+1}$
  is defined similarly, and reconstructs $G$ out of $t_G$ recursively
  as follows:

  \begin{itemize}
  \item Let $r$ be the root of $t_G$; it can be identified as the only
    vertex that does not have a parent. Remove $r$ from the structure,
    thus turning $t_G$ into a forest $t'_G$, where the roots of $t'_G$
    are children of $r$ in $t_G$.
  
  \item Using function extension, add a function $f$ which maps every
    vertex $v$ to its unique root ancestor in $t'_G$. This can be done
    by taking $f$ to be the $d$-fold composition of the parent
    function of $t'_G$ with itself (assuming each root points to
    itself, which can be easily interpreted).
  
  \item Copy all the roots of trees in $t'_G$ and let $X$ be the set
    of those copies. Construct a function $f'\colon V(t'_G)\to X$ that
    maps each vertex $v$ to the copy of $f(v)$.  Observe that $f'$
    defines a bundling of the trees of $t'_G$.
  
  \item Apply the transduction $\interp B_d^\star$ obtained from
    \cref{lem:parallel} to the above bundling.  This yields a bundling
    of the family of connected components of $G'$, where $G'$ is
    obtained from $G$ by flipping the adjacency within $W_{0}$.

  \item Forgetting all elements of the structure apart from the
    bundled connected components of $G'$ yields the graph $G'$.
    Construct the graph $G$ by flipping the adjacency inside the set
    $W_{0}$.  Note here that since the remaining vertices are exactly
    the leaves of the original tree $t_G$, the predicate $W_{0}$ is
    still carried by them.  Finally, clean the structure from all
    unnecessary predicates.
  \end{itemize}
  It is straightforward to see that transductions $\interp T_d$ and
  $\interp B_d$ satisfy all the requested properties.  This concludes
  the proof of \cref{lem:key}.
\end{proof}

\subsection{\cref{pro:bi-def} for classes of with low shrubdepth covers}\label{sec:key-general}
We now prove \cref{pro:bi-def} in the general case. As noted earlier,
this will finish the proof of the right-to-left implication in
\cref{thm:main}.
\begin{proof}[of \cref{pro:bi-def}]
  Let $\CCC$ be a class of graphs with low shrubdepth covers.  We fix
  a finite $2$-cover~$\cal U$ of $\CCC$ such that $\CCC[\cal U]$ has
  bounded shrubdepth. Let $N=\sup\set{|\cal U_G|\colon G\in \CCC}$,
  and for $G\in \CCC$ let~$\wh G$ be the extension of $G$ by unary
  predicates $U_1,\ldots,U_N$ such that
  $\set{U_1,\ldots,U_N}=\cal U_G$. Let
  $\wh \CCC=\set{\wh G\colon G\in \CCC}$. Then the class
  $\cal B=\wh\CCC[\cal U]$ has bounded~shrubdepth.

  Apply \cref{lem:key} to the class $\CCC[\cal U]$, yielding almost
  quantifier-free transductions $\interp T$ and~$\interp B$.  It is
  easy to construct an almost-quantifier free transduction
  $\interp S'$ such that for $G\in \CCC$, the structure
  $\interp S'(\wh G)$ is the union of the trees
  $T_U\in \interp T(G[U])$, one tree per each $U\in \cal U_G$, where
  the union is disjoint apart from the vertices which belong to $V(G)$
  (the leaves of the trees).  Indeed, we process $U_1,\ldots,U_N$ in
  order, and for each consecutive $U_i$ we apply the
  transduction~$\interp T$ to $G[U_i]$, appropriately modifying all
  its atomic operations so that the elements outside of $U_i$ are
  ignored and kept intact.  Recall all the constructed trees have
  depth bounded by a constant, say $d$.

  Now obtain $\interp S$ from $\interp S'$ by precomposing with a
  sequence of unary lifts introducing the predicates $U_1,\ldots,U_N$,
  and appending the following operations.  First, using extension
  operations introduce unary predicates $D_{i,\ell}$ for
  $i\in \set{1,\ldots,N}$ and $\ell\in \set{0,1,\ldots,d}$ such that
  $D_{i,\ell}$ selects nodes at depth $\ell$ in the tree
  $T_{U_i}$. Next, using an extension operation that introduces an
  adjacency relation binding every pair of elements $u,v$ such that
  $f(u)=v$ for some function $f$ in the signature (the parent
  functions).  Finally, use a sequence of reduct operations which drop
  all functions and non-unary relations from the signature, apart from
  adjacency.  Thus every output of $\interp S$ is a colored graph.
   
  Let $\FFF=\interp S(\CCC)$.  By \cref{lem:lsc}, $\FFF$ has low
  shrubdepth covers. Furthermore, each graph $H\in \interp S(G)$ for
  some $G\in \CCC$ is the union of at most $N$ trees, hence $H$ is
  $N$-degenerate and in particular excludes the biclique $K_{N+1,N+1}$
  as a subgraph.
   
  Hence by \cref{lem:bicliques} we infer that $\interp S(\CCC)$ has
  low treedepth covers, so by \cref{lem:ltd-lsd-covers},
  $\interp S(\CCC)$ is a class of bounded expansion.
   
  We are left with constructing a deterministic almost quantifier-free
  transduction $\interp I$ satisfying $\interp I(\interp S(G))=\{G\}$.
  This transduction should take on input a graph $H\in \interp S(G)$
  and turn it back to $G$.  The vertex set of $H$ consists of $V(G)$
  and trees $T_U$ for $U\in \cal U_G$, each built on top of the subset
  $U$ of $V(G)$ and of depth at most $d$.  Using predicates
  $D_{i,\ell}$ it is easy to use a sequence of quantifier-free
  function extension operations to construct, for each
  $U\in \cal U_G$, the parent function of $T_U$, thus turning the
  substructure induced by the nodes of $T_U$ back into~$T_U$.
  Similarly as before, it is now straightforward to construct a
  transduction $\interp I'$ that applies the transduction $\interp B$
  to each colored tree $T_U$, thus turning the set of its leaves into
  $G[U]$.  Since~$\cal U$ was a $2$-cover, for every edge $e$ of $G$
  there exists $U\in \cal U_G$ that contains both endpoints of~$e$.
  Hence, applying $\interp I'$ to the current structure recovers the
  graph $G$; this concludes the construction of $\interp I$. Note that
  $\interp I$ is deterministic almost quantifier-free.
\end{proof}


%% file: qe-covers.tex
\section{Proof of \cref{pro:qe-be0}\quad\normalfont\itshape\normalsize(quantifier elimination for classes of bounded expansion)}\label{sec:qe}
In this section we prove \cref{pro:qe-be0}, which we repeat for convenience. 
\setcounter{tmp}{\thetheorem}
\setcounterref{theorem}{pro:qe-be0}
\addtocounter{theorem}{-1}
\begin{proposition}\label{pro:qe-be}
  Let $\CCC$ be a class of graphs of bounded expansion and let
  $\interp I$ be a transduction.  Then $\interp I$ is equivalent to an
  almost quantifier-free transduction $\interp J$ on $\CCC$.
\end{proposition}
\setcounter{theorem}{\thetmp}

We note that \cref{pro:qe-be} is a strengthening of similar statements
provided by Dvo\v{r}\'ak et al.~\cite{DKT2} and of Grohe and
Kreutzer~\cite{Grohe2011}, and could be derived by a careful analysis
of their proofs, and by using the \cref{lem:star} below.

For a graph $G$ and a partial function $f\from V(G)\partto V(G)$, we
say that $f$ is \emph{guarded} by~$G$ if for every vertex in the
domain of $f$ is mapped to itself or to its neighbor.

\begin{lemma}[$\star$]\label{lem:star}
  Let $\CCC$ be a class of graphs which has $2$-covers of bounded
  treedepth, and for each $G\in \CCC$, let $\wh G$ be the graph $G$
  extended by a partial function $f\from V(G)\partto V(G)$ which is
  guarded by $G$. Then there is an almost quantifier-free transduction
  $\interp F$ using only unary lifts and a single function extension
  such that $\interp F(G)=\wh G$.
\end{lemma}

To derive \cref{pro:qe-be} from~\cite{DKT2}, one would need to prove
that the unary functions constructed in their proofs can be obtained
as compositions of guarded functions, and conclude using
\cref{lem:star}.  Rather then doing that, below we provide a
self-contained proof of \cref{pro:qe-be}, which we also believe is
simpler than the existing proofs, among other reasons, thanks to the
notion of covers. In \cref{sec:dkt} we outline how the result of \DKT
can be deduced from our proof.

\medskip

We will use the following restricted form of transductions.  A
\emph{faithful transduction} is a transduction which does not use
copying and restrictions.  A \emph{guarded transduction} is a faithful
transduction which given a structure $\str A$, produces a structure
whose Gaifman graph is a subgraph of the Gaifman graph of $\str A$.
In the following lemmas, we identify a first-order formula
$\phi(\bar x)$ with the transduction which inputs a structure $\str A$
and outputs~$\str A$ extended with a single relation, consisting of
those tuples $\bar a$ which satisfy $\phi(\bar x)$ in $\str A$ (this
transduction is a composition of an extension operation followed by a
sequence of reduct operations which drop all the symbols from the
input structure).

\begin{lemma}\label{lem:qe-formulas}
  Let $\phi(\bar x)$ be a first-order formula and let $\CCC$ be a
  class of graphs of bounded expansion.  Then there is a guarded
  transduction $\interp I$ which adds unary function and relation
  symbols only, and a quantifier-free formula $\phi'(\bar x)$, such
  that $\phi$ is equivalent to $\interp I; \phi'$ on $\CCC$.
\end{lemma}
Before proving \cref{lem:qe-formulas}, we first show how to conclude
\cref{pro:qe-be} using it.

\begin{proof}[of \cref{pro:qe-be}]
  For simplicity we assume that the signature produced by $\interp I$
  consists of one relation~$P$; lifting the proof to signatures
  containing more relation and function symbols is immediate.  By
  \cref{lem:normal}, we may express $\interp I$ as
  $$\interp I=\interp L;\interp C;\interp E;\interp X;\interp R,$$
  where
  \begin{itemize}
  \item $\interp L$ is a sequence of unary lifts,
  \item $\interp C$ is a sequence of copying operations,
  \item $\interp E$ is a single extension operation introducing the
    final relation $P$ using some formula $\phi(\bar x)$,
  \item $\interp X$ is a single universe restriction operation using
    some formula $\psi(x)$ that does not use symbol $P$, and
  \item $\interp R$ is a sequence of reduct operations that drop all
    relations and functions apart from~$P$.
  \end{itemize}
  From \cref{lem:lex} it follows that the class
  $\interp C(\interp L(\CCC))$ of colored graphs is a class of bounded
  expansion, and therefore, we may apply \cref{lem:qe-formulas} to it,
  and to the formulas $\phi(\bar x)$ and $\psi(x)$ considered above.

  Using \cref{lem:qe-formulas} we replace the formulas $\phi(\bar x)$
  and $\psi(x)$ by quantifier-free formulas, at the cost of
  introducing additional guarded transductions which introduce unary
  function and relation symbols.  Using \cref{lem:star}, every such
  transduction is equivalent to an almost quantifier-free
  transduction.  Hence, the transductions $\interp E$ and $\interp X$
  can be replaced in $\interp I$ by almost quantifier-free
  transductions, yielding an almost quantifier-free transduction
  $\interp J$ that is equivalent to $\interp I$ on $\CCC$.
\end{proof}

As explained, \cref{pro:qe-be} together with \cref{pro:bi-def} yields
\cref{thm:qe-lsc}. It remains to prove \cref{lem:qe-formulas}.
Similarly as in~\cite{DKT2,Grohe2011}, we first prove the statement
for classes of colored forests of bounded depth:

\begin{lemma}[$\star$]\label{lem:qe-trees0}
  Let $\phi(\bar x)$ be a first-order formula and let $\FFF$ be a
  class of colored rooted forests of bounded depth.  Then there is a
  transduction $\interp I_\phi$ which, given a rooted forest
  $F\in \FFF$ extends it by the parent function of $F$ and some unary
  predicates, and there exists a quantifier-free formula
  $\phi'(\bar x)$ such that $\phi$ is equivalent to
  $\interp I_\phi;\phi'$ on $\FFF$.
\end{lemma}
Let us remark that the presented proof of \cref{lem:qe-trees0} is
based on the automata approach and is conceptually different from the
ones used in~\cite{DKT2,Grohe2011}. Note that the
transduction~$\interp I_\phi$ produced in \cref{lem:qe-trees0} is in
particular a guarded transduction, since the parent of a vertex in a
forest is in particular a neighbor of that vertex.

\medskip The next step is to lift \cref{lem:qe-trees0} to classes of
structures of bounded treedepth.  We first observe that classes of
bounded treedepth are bi-definable with classes of forests of bounded
depth, using almost quantifier-free transductions. This result is
similar, but much simpler to prove than \cref{lem:key}, which is an
analogous statement for classes of bounded shrubdepth.

\begin{lemma}\label{lem:dfs}
  Let $\CCC$ be a class of structures of bounded treedepth. There is a
  pair of faithful transductions $\interp T$ and $\interp C$ and a
  class $\FFF$ of colored rooted forests of bounded depth such that
  $\interp T(\CCC)\subset \FFF$, $\interp C(\FFF)\subset \CCC$ and
  $\interp C(\interp T(\str A))=\set {\str A}$ for $\str A\in
  \CCC$.
  Moreover, the transduction $\interp T$ is guarded, and $\interp C$
  is deterministic almost quantifier-free.
\end{lemma}
\begin{proof}
  We follow the well-known encoding of structures of bounded treedepth
  inside colored forests, where a structure $\str A\in\CCC$ is encoded
  in a depth-first search forest of its Gaifman graph, as follows.
  
  A \emph{depth first-search} (DFS) forest of a graph $G$ is a rooted
  forest $F$ which is a subgraph of $G$, such that every edge of $G$
  connects an ancestor with a descendant in $F$.

  It is known that a graph $G$ of treedepth at most $d$ has a DFS
  forest of depth at most $2^d$.  If $\str A$ is a structure over a
  fixed signature $\Sigma$, $G$ is its Gaifman graph and $F$ is a DFS
  forest of $G$ of depth $2^d$, then $\str A$ can be encoded in $F$
  using a bounded number of additional unary predicates by labeling
  every node $v$ of $F$ by the isomorphism type of the substructure of
  $\str A$ induced by $v_1,\ldots,v_t$, where $v_1,\ldots,v_t$ are the
  nodes on the path from a root of $F$ to $v$, $v=v_t$ and $t\le 2^d$.
  The number of used unary predicates depends only on the signature
  $\Sigma$ and $d$.

  If $\CCC$ be a class of structures of treedepth at most $d$, then
  the transduction $\interp T$, given a structure $\str A\in \CCC$
  outputs a DFS forest $F$ of the Gaifman graph of $\str A$ of depth
  at most $2^d$, extended with unary predicates encoding $\str A$, as
  described above.  The structure $\str A$ can be recovered from $F$
  (together with the unary predicates) using a deterministic almost
  quantifier-free transduction, which first introduces the parent
  function, and then uses a quantifier-free formula to determine the
  quantifier-free type of a tuple of vertices.
\end{proof}

Using \cref{lem:dfs} we easily lift the quantifier-elimination result
from forests of bounded depth to classes of low treedepth.
\begin{lemma}\label{lem:qe-btd}
  Let $\phi(\bar x)$ be a first-order formula and let $\CCC$ be a
  class of structures of bounded treedepth.  Then there is a guarded
  transduction $\interp I_\phi$ and a quantifier-free formula
  $\phi'(\bar x)$ such that $\phi$ is equivalent to
  $\interp I_\phi;\phi'$ on $\CCC$.
\end{lemma}
\begin{proof}
  Let $\interp C,\interp T$ and $\FFF$ be as in \cref{lem:dfs}.  Since
  $\interp C(\interp T(\str A))=\set{\str A}$ and $\interp C$ is
  deterministic, there is a formula $\psi(\bar x)$ such that $\phi$ is
  equivalent to $\interp T;\psi$ on $\CCC$.  Now, apply
  \cref{lem:qe-trees0} to the class $\FFF$ and the formula
  $\psi(\bar x)$, yielding a guarded transduction $\interp J$ and a
  quantifier-free formula $\psi'(\bar x)$, such that $\psi$ is
  equivalent to $\interp J;\psi'$ on $\FFF$.  By composition, $\phi$
  is equivalent to $\interp T;\interp J;\psi'$ on $\CCC$.  Note that
  $\interp T;\interp J$ is a guarded transduction, since $\interp T$
  and $\interp J$ are such. This proves the lemma.
\end{proof}

%
%
%

Finally, we lift the quantifier elimination procedure to classes with
low shrubdepth covers using \cref{lem:dep} and a reasoning very
similar to the proof of \cref{lem:lsc}.  Again, conceptually this lift
is exactly what is happening in~\cite{DKT2,Grohe2011}, however, our
approach based on covers makes it quite straightforward.  The key
observation is encapsulated in the following lemma.
\begin{lemma}\label{lem:qf-covers}
  Let $\DDD$ be a class of structures with unary relation and function
  symbols only, and let $\phi(\bar x)$ be a quantifier-free formula
  with $p$ free variables, involving $c$ distinct terms.  Then there
  is a quantifier-free formula $\phi'(\bar x)$ such that following
  conditions are equivalent for a structure $\str A\in\DDD$, a
  $c\cdot p$-cover $\cal U_{\str A}$ of the Gaifman graph of $\str A$,
  and a $p$-tuple $\bar a$ of elements of $\str A:$
  \begin{enumerate}
  \item $\str A, \bar a\models \phi(\bar x)$,
  \item there is some $U\in \cal U_G$ containing $\bar a$ such that
    $\str A[U],\bar a\models\phi'(\bar x)$.
  \end{enumerate}
\end{lemma}
\begin{proof}
  We first consider the special case when $\phi(\bar x)$ is an atomic
  formula.
	%
  Each term $t$ occurring in $\phi(\bar x)$ defines a partial function
  $t_{\str A}\from V(\str A)\partto V(\str A)$ on a given
  structure~$\str A$, in the natural way.  Let $\cal T$ denote the set
  of terms occurring in $\phi(\bar x)$. By assumption,
  $|\cal T|\le c$.  For a tuple $\bar a=(a_1,\ldots,a_p)$ of elements
  of a structure $\str A$, denote by $\cal T_{\str A}(\bar a)$ the set
  $\setof{t_{\str A}(a_i)}{t\in \cal T, 1\le i\le p}$.  Then
  $|\cal T_{\str A}(\bar a)|\le c\cdot p$.
	
  Since $\phi(\bar x)$ is an atomic formula, for any $p$-tuple
  $\bar a$ of elements of $\str A$ and any set $U\subset V(\str A)$
  containing $\cal T_{\str A}(\bar a)$ we have the following
  equivalence:
  $$\str A,\bar a\models\phi(\bar x)\iff 
  \str A[U],\bar a\models\phi(\bar x).$$

  Take $\phi'(\bar x)=\phi(\bar x)$.  The equivalence of the two items
  then follows by assumption that $\cal U_G$ is a $p\cdot c$-cover of
  $\str A$, so for every $\bar a$, there is some set $U\in \cal U_G$
  containing $\cal T_{\str A}(\bar a)$.

  To treat the general case of a quantifier-free formula, we take
  $\phi'(\bar x)$ to be a conjunction of $\phi(\bar x)$ and a formula
  which verifies that all the values in $\cal T_{\str A}(\bar a)$ are
  defined. We leave the details to the reader.
\end{proof}

We are ready to prove \cref{lem:qe-formulas}.

\begin{proof}[of \cref{lem:qe-formulas}]
  The proof proceeds by induction on the structure of the formula
  $\phi(\bar x)$. In the base case, $\phi(\bar x)$ is a
  quantifier-free formula, so we may take $\interp I$ to be the
  identity transduction.


  In the inductive step, we consider two cases.  If $\phi(\bar x)$ is
  a boolean combination of simpler formulas, then the statement
  follows immediately from the inductive assumption.  The interesting
  case is when $\phi(\bar x)$ is of the form
  $\exists y.\psi(\bar x,y)$, for some formula $\psi(\bar x,y)$. We
  consider this case below. Denote by $p$ the number of free variables
  in the formula~$\psi(\bar x,y)$.
	
  Apply the inductive assumption to the formula $\psi(\bar x,y)$,
  yielding a guarded transduction $\interp I_\psi$ and a formula
  $\psi'(\bar x,y)$. Let $c$ be the number of distinct terms
  (including subterms) appearing in the formula $\psi'(\bar x,y)$.
  Let $\DDD=\interp I_\psi(\CCC)$. Note that every structure in $\DDD$
  has unary function and relation symbols only, and is guarded by some
  graph in~$\CCC$.  By \cref{lem:lsc}, we can pick a finite
  $c\cdot p$-cover $\cal U$ of $\CCC$, so that the class
  $\CCC[\cal U]$ has bounded treedepth. As $\interp I_\psi$ is
  guarded, it follows that also the class $\DDD[\cal U]$ has bounded
  treedepth.
 
  Apply \cref{lem:qf-covers} to $\DDD$ and $\psi'(\bar x,y)$, yielding
  a formula $\psi''(\bar x,y)$ such that for every graph $G\in \CCC$,
  $p$-tuple of vertices $(\bar a,b)$ and the $c\cdot p$-cover
  $\cal U_{G}$ of $G\in \CCC$, the following equivalences hold:
  \begin{align*}
    G,\bar a,b\models\psi(\bar x,y) & \iff
                                      \interp I_{\psi}(G),\bar a,b\models\psi'(\bar x,y)\\ &\iff
                                                                                             \interp I_{\psi}(G)[U],\bar a,b\models\psi''(\bar x,y)\ \textit{for some $U\in \cal U_G$ containing $\bar a,b$}.    
  \end{align*}

  Apply \cref{lem:qe-btd} to the class $\DDD[\cal U]$ and the formula
  $\exists y.\psi''(\bar x,y)$, yielding a guarded transduction
  $\interp F$ and quantifier-free formula $\rho(\bar x)$ such that for
  every $\str A\in \DDD[\cal U]$ and tuple
  $\bar a\in V(\str A)^{|\bar x|}$,
  $$\str A,\bar a\models \exists y.\psi''(\bar x,y)
  \iff \interp F(\str A),\bar a\models \rho(\bar x).$$
  
  \begin{claim}\label{cl1}For each graph $G\in\CCC$ and tuple
    $\bar a\in V(H)^{|\bar x|}$, the following conditions are
    equivalent:
    \begin{enumerate}
    \item $G,\bar a\models\exists y.\psi(\bar x,y)$,
    \item there is some $U\in \cal U_G$ containing $\bar a$ such that
      $\interp F(\interp I_\psi(G)[U]),\bar a\models \rho(\bar x)$.
    \end{enumerate}
	  
  \end{claim}
  \begin{clproof}We have the following equivalences:
    \begin{align*}
      G,\bar a\models\exists y.\psi(\bar x,y) &\iff
                                                G,\bar a,b\models\psi(\bar x,y)\textit{\ for some $b\in V(G)$}\\&\iff  
%
      \interp I_\psi(G)[U],\bar a,b\models\psi''(\bar x,y)\textit{\ for some  $U\in \cal U_G$ containing $\bar a,b$}\\&\iff      
                                                                                                                        \interp I_\psi(G)[U],\bar a\models\exists y.\psi''(\bar x,y)\textit{\ for some  $U\in \cal U_G$ containing $\bar a$}\\&\iff
                                                                                                                                                                                                                                                \interp F(\interp I_\psi(G)[U]),\bar a\models\rho(\bar x)\textit{\ for some  $U\in \cal U_G$ containing $\bar a$}.
    \end{align*}
    This proves the claim.
  \end{clproof}

  Let $N=\sup\set{|\cal U_G|\colon G\in \CCC}$.  For each graph
  $G\in \CCC$, fix an enumeration $U_1,\ldots,U_N$ of the cover
  $\cal U_G$.

  \begin{claim}\label{cl2}
    There is a guarded transduction $\interp F'$ and quantifier-free
    formulas $\rho_1(\bar x),\ldots,\rho_N(\bar x)$ such that given a
    graph $G\in\CCC$, a number $i\in\set{1,\ldots,N}$ and a tuple
    $\bar a$ of elements of~$U_i$,
	
    $$\interp F'(G),\bar a\models \rho_i(\bar x)\iff 
    \interp F(\interp I_\psi (G)[U_i]),\bar a\models \rho(\bar x).  $$
  \end{claim}
  \begin{clproof}
    We construct a guarded transduction $\interp F'$ which, given a
    graph $G\in \CCC$, first applies the guarded transduction
    $\interp I_\psi$, then introduces unary predicates marking the
    sets $U_1,\ldots,U_N$, and then, for each such unary predicate
    $U_i$, applies to the structure $\interp I_\psi(G)[U_i]$ the
    transduction $\interp F$, modified so that each function symbol
    $f$ is replaced by a new function symbol $f^i$.
    
    Then the formula $\rho_i(\bar x)$ is obtained from the formula
    $\rho(\bar x)$, by replacing each function symbol $f$ by the
    function symbol $f^i$.
  \end{clproof}

Combining \cref{cl1} and \cref{cl2} we get the following equivalence:
$$\interp F'(G),\bar a\models \bigvee_{i=1}^N \rho_i(\bar x)\iff 
G,\bar a\models \phi(\bar x),$$
concluding the inductive step. This finishes the proofs of \cref{lem:qe-formulas} and \cref{pro:qe-be}.\\
\mbox{}
\end{proof}


%

\subsection{Effectivity}\label{sec:dkt}
As a side remark, we note that we can easily derive the result of
\DKT, by observing that the above proof of \cref{lem:qe-formulas} is
effective, and can be leveraged to construct a transduction
$\interp I$ which is a linear time computable function.

%
%
%
%

We say that a transduction $\interp I$ is a \emph{linear time}
transduction if there is an algorithm which, given a structure
$\str A$ as input, produces some structure
$\str B\in \interp I(\str A)$ in linear time. Here, the structure
$\str A$ is represented using the adjacency list representation, i.e.,
for a colored graph, the size of the description is linear in the sum
of the number of vertices and the number of edges in the graph.

\medskip We show the following, effective variant of
\cref{lem:qe-formulas}.
\begin{lemma}\label{lem:qe-formulas-eff}
  Let $\phi(\bar x)$ be a first-order formula and let $\CCC$ be a
  class of graphs of bounded expansion.  Then there is a guarded
  transduction $\interp I$ which adds unary function and relation
  symbols only, and a quantifier-free formula $\phi'(\bar x)$, such
  that $\phi$ is equivalent to $\interp I; \phi'$ on $\CCC$.
  Moreover, $\interp I$ is a linear time transduction.
\end{lemma}


\begin{proof}
  To prove \cref{lem:qe-formulas-eff}, we observe that the
  transduction $\interp I$ in \cref{lem:qe-formulas} is a linear time
  transduction.  The proof follows by tracing the proof of
  \cref{lem:qe-formulas}, and observing the following.
  \begin{enumerate}
  \item In \cref{lem:qe-trees0}, the constructed transduction
    $\interp I$ is a linear time transduction. This is because the
    transduction only adds the parent function (which is clearly
    linear-time computable, given a rooted forest) and some unary
    predicates, each of which can be computed in linear time, since
    each unary predicate is produced by running a deterministic
    threshold tree automaton on the input tree.
		
  \item In \cref{lem:dfs}, the transduction $\interp T$ is a linear
    time transduction, since it amounts to running a depth-first
    search on the input graph.
		
  \item In \cref{lem:qe-btd}, the produced transduction
    $\interp J=\interp T;\interp J$ is a linear time transduction, as
    a composition of two linear time transductions.
		
  \item In the proof of \cref{lem:qe-formulas}, the nontrivial step is
    in the inductive step, in the case of an existential
    formula. 
    In this case, the constructed transduction $\interp F'$ is a
    linear time transduction, assuming $\CCC$ has bounded expansion,
    as $\interp F'$ amounts to introducing unary predicates denoting
    the elements of a cover $\cal U_G$, and applying transductions
    $\interp I_\psi$ and $\interp F$ which are linear time
    transductions, respectively, by the inductive assumption, and by
    the effective version of \cref{lem:qe-btd} discussed above.
		
    We note that if $\CCC$ has bounded expansion then for any fixed
    $p\ge 0$ there is a finite $p$-cover $\cal U$ of $\CCC$ of bounded
    treedepth such that $\cal U_G$ can be computed from a given
    $G\in \CCC$ in time $f(p)\cdot |V(G)|$, for some function $f$
    depending on $\CCC$ (the function $f$ may not be computable). To
    compute $\cal U_G$, we may first compute a $g(p)$-treepdepth
    coloring of $G$ for some function $g$ (as required in the proof of
    \cref{lem:ltd-lsd-covers}) and observe that it can be converted to
    a cover in linear time, as in the proof of
    \cref{lem:ltd-lsd-covers}. A $p$-treedepth coloring can be
    computed in linear time, cf.~\cite{POMNII, dkt, Sparsity}.
  \end{enumerate}

\end{proof}

%% file: effectivity_short.tex
\section{Algorithmic aspects}
\label{sec:alg}

In this section we give a preliminary result about efficient
computability of transductions on classes with structurally bounded
expansion. When we refer to the size of a structure in the algorithmic
context, we refer to its total size, i.e., the sum of its universe
size and the total sum of sizes of tuples in its relations.

Call a class \(\CCC\)
of graphs of structurally bounded expansion {\em{efficiently
    decomposable}} if there is a finite $2$-cover $\cal U$ of $\CCC$
and an algorithm that, given a graph $G\in \CCC$, in linear time
computes the cover $\cal U_G$ and for each $U\in\cal U_G$, an
SC-decomposition $S_U$ of depth at most $d$ of the graph $G[U]$, for
some constant $d$ depending only on $\CCC$. Our result is as follows.

\begin{theorem}\label{thm:algo} Suppose $\interp J$ is a deterministic
  transduction and $\CCC$ is a class of graphs that has structurally
  bounded expansion and is efficiently decomposable. Then given a
  graph $G\in \CCC$, one may compute $\interp J(G)$ in time linear in
  the size of the input plus the size of the output. \end{theorem}

We remark that instead of efficient decomposability we could assume
that the $2$-cover~$\cal U_G$ of a graph $G$ and corresponding
SC-decompositions for all $U \in \cal U_G$ is given together with $G$
as input. If only the cover is given but not the SC-decompositions, we
would obtain cubic running time because bounded shrubdepth implies
bounded cliquewidth and we can compute an approximate clique
decomposition in cubic time~\cite{oum2008approximating}. Then,
SC-decompositions of small height are definable in monadic
second-order logic, and hence they can be computed in linear time
using the result of Courcelle, Makowski and
Rotics~\cite{Courcelle2000a}.

Observe that the theorem implies that we can efficiently evaluate a
first-order sentence and enumerate all tuples satisfying a formula
$\varphi(x_1, \ldots, x_k)$ on the given input graph, since this
amounts to applying the theorem to a transduction consisting of a
single extension operation. This strengthens the analogous result of
Kazana and Segoufin~\cite{Kazana2013} for classes of bounded
expansion.

\begin{proof}[sketch] We will make use of transductions $\interp S$
  and $\interp I$ constructed in the proof of
  \cref{pro:bi-def}. Recall that $\interp S(\CCC)$ is a
  class of colored graphs of bounded expansion, $\interp I$ is
  deterministic, and $\interp I(\interp S(G))=\{G\}$ for each
  $G\in \CCC$. Observe that $\interp J$ is equivalent to
  $\interp S;\interp I;\interp J$ on~$\CCC$. Defining $\interp K$ as
  $\interp I;\interp J$, we get that
  $\interp J(G)=\interp K(\interp S(G))$ for $G\in \CCC$. Moreover,
  since~$\interp I$ is deterministic, it follows that $\interp K$ is
  deterministic.

  Let $G\in\CCC$ be an input graph. By efficient decomposability of
  $\CCC$, in linear time we can compute a cover $\cal U_G$ of $G$
  together
  with an SC-decomposition $S_U$ of depth at most~$d$ of $G[U]$, for
  $U\in\cal U_G$. Each $S_U$ is a colored tree, and by the
  construction described in the proof of \cref{pro:bi-def},
  the trees $S_U$ for $U\in\cal U_G$, glued along the leaves form a
  structure belonging to~$\interp S(G)$. As
  $\interp J(G)=\interp K(\interp S(G))$, it suffices to apply the
  enumeration result of Kazana and Segoufin for classes of bounded
  expansion~\cite{Kazana2013} to the colored graph~$\interp S(G)$ and
  to all formulas occurring in the transduction $\interp K$.
\end{proof}

%% file: conclusion_full.tex
\section{Conclusion}\label{sec:conclusion}

In this paper we have provided a natural combinatorial
characterization of graph classes that are first-order transductions
of bounded expansion classes of graphs.  Our characterization
parallels the known characterization of bounded expansion classes by
the existence of low treedepth decompositions, by replacing the notion
of treedepth by shrubdepth. We believe that we have thereby taken a
big step towards solving the model-checking problem for first-order
logic on classes of structurally bounded
expansion.

On the structural side we remark that transductions of bounded
expansion graph classes are just the same as transductions of classes
of structures of bounded expansion (i.e., classes whose Gaifman graphs
or whose incidence encodings have bounded expansion). On the other
hand, it remains an open question to characterize classes of
relational structures, rather than just graphs, which are
transductions of bounded expansion classes. We are lacking the
analogue of \cref{lem:key}; the problem is that within the proof we
crucially use the characterization of shrubdepth via SC-depth, which
works well for graphs but is unclear for structures of higher arity.

Finally, observe that classes of bounded expansion can be
characterized among classes with structurally bounded expansion as
those which are bi-clique free. It follows, that every monotone (i.e.,
subgraph closed) class of structurally bounded expansion has bounded
expansion. Exactly the same statement holds characterizing bounded
treedepth among bounded shrubdepth, and the second item holds for
treewidth vs cliquewidth.  In particular, for monotone graph classes
all pairs of notions collapse.

We do not know how to extend our results to nowhere dense classes of
graphs, mainly due to the fact that we do not know whether there
exists a robust quantifier-elimination procedure for these graph
classes.


%% file: app-normal.tex
\section{Normalization lemmas for transductions}\label{app:prelims-omitted}

In this section we give proofs omitted from \cref{sec:prelims-logic}.

\begin{proof}[of \cref{lem:normal} and of \cref{lem:normal-qf}]
  We give appropriate swapping rules that allow us to arrange the
  atomic operations comprising $\interp I$ into the desired normal
  form.

  We start with putting all the unary lifts at the front of the
  sequence. Observe that whenever an atomic operation is followed by a
  unary lift, then these two operations may be appropriately swapped.
  This is straightforward for all atomic operations apart from
  copying. For this last case, observe that copying followed by a
  unary lift introducing a unary predicate~$X$ is equivalent to a
  transduction that does the following. First, using unary lifts
  introduce two auxiliary unary predicates $X_1$ and $X_2$,
  interpreted to select vertices that are supposed to be selected by
  $X$ in the original universe, respectively in the copy of the
  universe. Then perform copying. Finally, use extension and reduct
  operations to appropriately interpret $X$ and drop predicates
  $X_1,X_2$.

  Having applied the above swapping rules exhaustively, the formula is
  rewritten into the form $\interp L;\interp I'$ where $\interp I'$
  does not contain any lifts. Observe that if $\interp I$ was almost
  quantifier-free, then $\interp I'$ is deterministic almost
  quantifier-free. This proves \cref{lem:normal-qf}.

  Next, we perform swapping within $\interp I'$ so that all copying
  operations are put at the front of the sequence of atomic
  operations. Again, it suffices to show that whenever an atomic
  operation is followed by copying, then the two operations may be
  swapped. For reducts this is obvious, while for extensions and
  restrictions one should modify the formula parameterizing the
  operation in a straightforward way to work on each copy separately.
  Thus we have rewritten $\interp I$ into the form
  $\interp L;\interp C;\interp I''$ where $\interp I''$ does not use
  lifts or copying.

  Now consider $\interp I''$. It is clear that all reduct operations
  can be moved to the end of the transduction, since it does not harm
  to have more relations in the structure. Next, we move all
  restriction operations to the end (before reduct operations) by
  showing that each restriction operation can be swapped with any
  extension or function extension operation. Suppose that the
  restriction is parameterized by a unary formula $\psi$, and it is
  followed by an extension operation (normal or function), say
  parameterized by a formula $\varphi$. Then the two operations may be
  swapped provided we appropriately {\em{relativize}} $\varphi$ as
  follows: add guards to all quantifiers in $\varphi$ so that they run
  only over elements satisfying $\psi$, and for every term $\tau$ used
  in $\varphi$ add guards to check that all the intermediate elements
  obtained when evaluating $\tau$ satisfy $\psi$.

  Applying these swapping rules exhaustively rewrites $\interp I''$
  into the form $\interp I''';\interp X';\interp R$,
  where~$\interp I'''$ is a sequence of extension and function
  extension operations, $\interp X'$ is a sequence of restriction
  operations, and $\interp R$ is a sequence of reduct operations. We
  now argue that~$\interp X'$ can be replaced with a single
  restriction operation $\interp X$. It suffices to show how to do
  this for two consecutive restriction operations, say parameterized
  by $\psi_1$ and $\psi_2$, respectively. Then we may replace them by
  one restriction operation parameterized by $\psi_1\wedge\psi_2'$,
  where $\psi_2'$ is obtained from $\psi_2$ by relativizing it with
  respect to $\psi_1$ just as in the previous paragraph.

  We are left with treating the extension and function extension
  operations within $\interp I'''$. Whenever a formula $\varphi$
  parameterizing some extension or function extension operation
  within~$\interp I'''$ uses a relation symbol $R$ introduced by some
  earlier extension operation within~$\interp I'''$, say parameterized
  by formula~$\varphi'$, then replace all occurrences of $R$ in
  $\varphi$ with~$\varphi'$. Similarly, if~$\varphi$ uses some
  function $f$ that was introduced by some earlier function extension
  operation within $\interp I'''$, say using formula $\varphi'(x,y)$,
  then replace each usage of $f$ in $\varphi$ by appropriatiely
  quantifying the image using formula $\varphi'(x,y)$. Perform the
  same operations on the formula parameterizing the restriction
  operation $\interp X$.

  Having performed exhaustively the operations above, formulas
  parameterizing all atomic operations in $\interp I''';\interp X$ use
  only relations and functions that appear originally in the structure
  or were added by $\interp L;\interp C$. Hence, all extension and
  function extension operations within $\interp I'''$ which introduce
  symbols that are later dropped in $\interp R$ can be simply removed
  (together with the corresponding reduct operation). It now remains
  to observe that all atomic operations within~$\interp I'''$ commute,
  so they can be sorted: first function extensions, then (normal)
  extensions. 
\end{proof}

%% file: bicliques.tex
\section{Proof of \cref{lem:bicliques}}

In this section we prove \cref{lem:bicliques}. One implication is
easy: it is known~\cite{Ganian2012} that every class of bounded
treedepth also has bounded shrubdepth, and moreover the
bi-clique~$K_{s,s}$ has treedepth $s+1$, so every class of bounded
treedepth excludes some bi-clique.

We need to prove the reverse implication: any class of bounded
shrubdepth that moreover excludes some bi-clique has bounded
treedepth. We will use the following well-known characterization of
classes of bounded treedepth (see~\cite[Theorem 13.3]{Sparsity}).

\begin{lemma}\label{lem:td-path} 
  A class of graphs $\CCC$ has bounded
  treedepth if and only if there exists a number $d\in \N$ such that
  no graph from $\CCC$ contains a path on more than $d$ vertices as a
  subgraph. 
\end{lemma}

By \cref{lem:td-path} and \cref{prop:sd_properties}(\ref{SD:3}), to
prove \cref{lem:bicliques} it is sufficient to prove the following.

\begin{lemma} 
  There exists a function
  $g\colon \N\times\N\times \N\to \N$ such that the following holds.
  For all integers $h,m,s\in \N$, if a graph $G$ does not contain the
  bi-clique $K_{s,s}$ as a subgraph and admits a connection model of
  height at most $h$ using at most $m$ labels, then $G$ does not
  contain any path on more than $g(h,m,s)$ vertices as a subgraph.
\end{lemma} 
\begin{proof}
  We proceed by induction on the height $h$. For $h=0$, only one-vertex
  graphs admit a connection model of height $0$, so we may set
  $g(0,m,s)=1$.

  For the induction step, suppose $G$ does not contain $K_{s,s}$ as a
  subgraph and admits a connection model $T$ of height $h\geq 1$ and
  using $m$ labels. Call two vertices $u$ and $v$ of $G$
  {\em{related}} if they are contained in the same subtree of $T$
  rooted at a child of the root of $G$, and {\em{unrelated}}
  otherwise. Whenever $u$ and $v$ are unrelated, their least common
  ancestor is the root of $T$, so whether they are adjacent depends
  solely on the pair of their labels.

  Let $P=(v_1,\ldots,v_p)$ be a path in $G$. A {\em{block}} on $P$ is
  a maximal contiguous subpath of~$P$ consisting of vertices that are
  pairwise related. Thus, $P$ breaks into blocks $B_1,\ldots,B_q$,
  appearing on $P$ in this order. Note that each block $B_i$ is a path
  that is completely contained in an induced subgraph of $G$ that
  admits a connection model of height $h-1$ and using~$m$ labels.
  Hence, by the induction hypothesis we have that each block $B_i$ has
  at most $g(h-1,m,s)$ vertices.

  For a non-last block $B_i$ (i.e. $i\leq q$), define the
  {\em{signature}} of $B_i$ as the pair of labels of the following two
  vertices: the last vertex of $B_i$ and of its successor on $P$, that
  is, the first vertex of $B_{i+1}$. The following claim is the key
  point of the proof.

  \begin{claim}\label{cl:single-sig} 
    For any signature, the number of
    non-last blocks with this signature is at most $4(s-1)$.
  \end{claim}
  \begin{proof}
\renewcommand{\qedsymbol}{\ensuremath{\lrcorner}}
  Let $\sigma=(\lambda_1,\lambda_2)$ be the signature in question and
  let $\mathcal{B}$ be the set of blocks with signature~$\sigma$;
  suppose for the sake of contradiction that $|\mathcal{B}|>4(s-1)$.
  Consider the following random experiment: independently color each
  subtree of $T$ rooted at a child of the root black or white, each
  with probability $1/2$. Call a block $B_i\in \mathcal{B}$
  {\em{split}} if the last vertex of~$B_i$ is white and the first
  vertex of $B_{i+1}$ is black. Since these two vertices are unrelated
  (by the maximality of~$B_i$), each block $B_i$ is split with
  probability~$1/4$, implying that the expected number of split blocks
  is $|\mathcal{B}|/4>s-1$. Hence, some run of the experiment yields a
  white/black coloring of subtrees rooted at children of the root of
  $T$ and a set $\mathcal{S}\subseteq \mathcal{B}$ of $s$ blocks that
  are split in this coloring.

  Let $u_1,\ldots,u_s$ be the last vertices of blocks from
  $\mathcal{S}$ and $v_1,\ldots,v_s$ be their successors on the path
  $P$, respectively. By assumption, all vertices $u_i$ have label
  $\lambda_1$ and all vertices~$v_i$ have label $\lambda_2$. Further,
  all vertices $u_i$ are white and all vertices $v_i$ are black,
  implying that~$u_i$ and~$v_j$ are unrelated for all $i,j\in [s]$.
  Since $u_i$ is unrelated and adjacent to $v_i$, it follows
  that~$u_i$ is adjacent to all vertices $v_j$, $j\in [s]$, as these
  vertices are also unrelated to $u_i$ and have the same label as
  $v_j$. We conclude that $u_1,\ldots,u_s$ and $v_1,\ldots,v_s$ form a
  bi-clique $K_{s,s}$ in $G$, a contradiction. 
\end{proof}

Since the number of possible signatures is $m^2$, by
\cref{cl:single-sig} we infer that the total number of blocks is at
most $4(s-1)m^2+1$. As we argued, each block has at most $g(h-1,m,s)$
vertices, implying $p\leq (4(s-1)m^2+1)\cdot g(h-1,m,s)$. As $P$ was
chosen arbitrarily, we may set
$$g(h,m,s)\coloneqq (4(s-1)m^2+1)\cdot g(h-1,m,s).$$
This concludes the inductive proof.
\end{proof}

%% file: app-sbe_covers.tex
\section{Proof of \cref{lem:ltd-lsd-covers}}

\begin{proof}[of \cref{lem:ltd-lsd-covers}]
  We will prove that a graph class $\CCC$ has low treedepth colorings
  if and only if it has low treedepth covers. The result then follows
  from \cref{thm:beltc}.
  
  We start with the left-to-right direction. Assume $\CCC$ has low
  treedepth colorings.  Then for every graph $G\in \CCC$ and $p\in \N$
  we may find a vertex coloring $\gamma\from V(G)\to [N]$ using~$N$
  colors where every $i \le p$ color classes induce in $G$ a subgraph
  of treedepth at most~$i$; here, $N$ depends only on $p$ and $\CCC$.
  Assuming without loss of generality that $N\ge p$, define a
  $p$-cover~$\cal U_G$ of size at most $\binom{N}{p}$ as follows:
  $\cal U_G=\set{\gamma^{-1}(X)\colon X\subset [N], |X|=p}$.
  Then~$\cal U=(\cal U_G)_{G\in \CCC}$ is a finite $p$-cover of $\CCC$
  of bounded treedepth.
  
  Conversely, suppose that every graph $G\in \CCC$ admits a $p$-cover
  $\cal U_G$ of size $N$ where $G[U]$ has treedepth at most $d$ for
  each $U\in \cal U_G$; here, $N$ and $d$ depend only on $p$
  and~$\CCC$.  Define a coloring $\chi\colon V(G)\to \cal P(\cal U_G)$
  as follows: for $v\in V(G)$, let $\chi(v)$ be the set of those
  $U\in \cal U_G$ for which $v\in U$.  Thus, $\chi$ is a coloring of
  $V(G)$ with $2^N$ colors.  Take any $p$ subsets
  $X_1,\ldots,X_p\subseteq \cal U_G$ such that
  $\chi^{-1}(X_i)\neq \emptyset$ for each $i\in [p]$. Arbitrarily
  choose any $x_i\in \chi^{-1}(X_i)$.  Since $\cal U_G$ is a $p$-cover
  of $G$, there exists $U\in \cal U_G$ such that
  $\{x_1,\ldots,x_p\}\subseteq U$.  Consequently, for each $i\in [p]$
  we have that $U\in X_i$, implying $\chi^{-1}(X_i)\subseteq U$.
  Hence $G[\chi^{-1}(\{X_1,\ldots,X_p\})]$ is an induced subgraph of
  $G[U]$, whereas the latter graph has treedepth at most $d$ by the
  assumed properties of $\cal U_G$.  We conclude that every $p$ color
  classes in $\chi$ induce a subgraph of treedepth at most $d$.
  
  It remains to refine this coloring so that we in fact obtain a
  coloring such that every at most $i\leq p$ color classes induce a
  subgraph of treedepth at most $i$.  As every~$p$ color classes
  in~$\chi$ induce a subgraph of treedepth at most $d$, we can fix for
  every~$p$ color classes~$I$ of $\chi$ a treedepth decomposition
  $Y_I$ of height at most $d$. We define the coloring $\xi$ such that
  every vertex~$v$ gets the color $\{(I, h_I) :$ $I$ is a subset of
  $p$ color classes containing $v$ and $h_I$ is the depth of $v$ in
  the decomposition $Y_I\}$. Note that since the number of colors of
  $\chi$ is finite, the number of colors used by $\xi$ is also finite.
  
  We now prove that in the refined coloring, any $i\leq p$ colors in
  $\xi$ have treedepth at most $i$. Fix any $i\leq p$ colors in $\xi$
  and denote the tuple of colors by $J$. As $\xi$ is a refinement of
  $\chi$, there exists a tuple $I$ of at most $p$ colors in $\chi$
  which contains all vertices of $G[J]$. Furthermore, the~$i$ selected
  colors of $J$ are contained in $i$ levels of the treedepth
  decomposition $Y_I$.  Taking the restriction of these $i$ levels
  yields a forest of height at most $i$, which is a witness that
  $G[J]$ has treedepth at most $i$.
\end{proof}

%% file: Treducible.tex
\section{Proofs of \cref{sec:def-ccs}}\label{app:connectome}

In this section we present the missing proofs of \cref{sec:def-ccs} as
well as a second proof for \cref{lem:go-right}.


\subsection{Guided and guidable functions}\label{sec:directable}


%

\begin{proof}[of \cref{lem:glue-directed}]
  For each connected component $C$ of $G$ we may find a guidance
  system $\Uu^C=\{U^C_1,\ldots,U^C_\ell\}$ that guides $g|_C$.  Since
  $g|_C$ is undefined for vertices outside of $C$, we may assume that
  $U^C_i\subseteq V(C)$ for each $i\in [\ell]$.  It follows that $g$
  is guided by the guidance system $\Uu=\{U_1,\ldots,U_\ell\}$ defined
  by setting $U_i$ to be the union of $U^C_i$ throughout connected
  components $C$ of $G$.
\end{proof}

\begin{proof}[of \cref{lem:glue-directable}]
  Let $\Uu_i$ be a guidance system of size at most $\ell$ that such
  that~$g_i$ is guided by~$\Uu_i$.  Then $\Uu=\bigcup_{i=1}^s \Uu_i$
  is a guidance system of size at most $\ell\cdot s$. It is easy to
  see that $\Uu$ guides the partial function $g$.
\end{proof}

%
%

\begin{proof}[of \cref{lem:local}]
  Let $\Uu$ be a guidance system of size at most $\ell$ such that
  $f_G$ is guided by $\Uu$.  For each vertex $x$ such that $f(x)$ is a
  neighbor of $x$, pick an arbitrary set $V(x)\in\Uu$ such that $f(x)$
  is the unique neighbor of $x$ in $V(x)$.

  We now present an almost quantifier-free transduction that
  constructs $f_G$.  First, for each $U\in \Uu$ use a unary lift to
  introduce a unary predicate that selects the vertices of $U$.  Next,
  introduce two unary predicates, $\textit{Null}$ and $\textit{Self}$,
  which select the vertices~$x$ such that~$f(x)$ is undefined or
  $f(x)=x$, respectively.  Finally, for each $V\in\Uu$ introduce a
  unary predicate ${G_V}$ that selects vertices $x$ with $V(x)=V$.
  Now, for each $U\in \Uu$, construct the partial function~$d_U$ which
  maps every vertex $x$ to its unique neighbor in $U$ (if it exists)
  using the function extension operation parameterized by the formula
  $E(x,y)\land U(y)$.  Finally, construct $f_G$ using the function
  extension operation parameterized by the formula $\alpha(x,y)$
  stating that $x\not\in\textit{Null}$ and either $x\in\textit{Self}$
  and $y=x$, or $x\in G_V$ and $y=d_U(x)$.
\end{proof}

\input{Treducible-mipi}

%% file: Treducible-mipi.tex
\subsection{Greedy proof of \cref{lem:go-right}}\label{app:greedy}

We now present the second proof of \cref{lem:go-right}.  As asserted
by \cref{lem:ind-paths}, graphs from a fixed class of bounded
shrubdepth do not admit arbitrarily long induced paths.  We need a
strengthening of this statement: classes of bounded shrubdepth also
exclude induced structures that roughly resemble paths, as made
precise next.

\begin{definition}
  Let $G$ be a graph. A {\em{quasi-path}} of length $\ell$ in $G$ is a
  sequence of vertices $(u_1,u_2,\ldots,u_\ell)$ satisfying the
  following conditions:
  \begin{itemize}
  \item $u_iu_{i+1}\in E(G)$ for all $i\in [\ell-1]$; and
  \item for every odd $i\in [\ell]$ and even $j\in [\ell]$ with
    $j>i+1$, we have $u_iu_j\notin E(G)$.
  \end{itemize}
\end{definition}

Note that in a quasi-path we do not restrict in any way the
adjacencies between $u_i$ and~$u_j$ when $i,j$ have the same parity,
or even when $i$ is odd and $j$ is even but $j<i-1$.  We now prove
that classes of bounded shrubdepth do not admit long quasi-paths; note
that since an induced path is also a quasi-path, the following lemma
actually implies \cref{lem:ind-paths}.

\begin{lemma}\label{lem:exc-quasi}
  For every class $\CCC$ of graphs of bounded shrubdepth there exists
  a constant $q\in \N$ such that no graph from $\CCC$ contains a
  quasi-path of length $q$.
\end{lemma}
\begin{proof}
  It suffices to prove the following claim.

\begin{claim}\label{cl:conn-mod-quasi}
  There exists a function $f\colon \N\times\N\to \N$ such that no
  graph admitting a connection model of height $h$ and using $m$
  labels contains a quasi-path of length larger than $f(h,m)$.
\end{claim}

The proof is by induction on $h$. Observe first that graphs admitting
a connection model of height $0$ are exactly graphs with one vertex,
hence we may set $g=f(0,m)=1$ for all $m\in \N$.

We now move to the induction step. Assume $G$ admits a connection
model $T$ of height $h\geq 1$ where $\lambda\colon V(G)\to \Lambda$ is
the corresponding labeling of $V(G)$ with a set $\Lambda$ consisting
of $m$ labels. Call two vertices $u,v\in V(G)$ {\em{related}} if in
$T$ they are contained in the same subtree rooted at a child of the
root of $T$; obviously this is an equivalence relation.  The least
common ancestor of two unrelated vertices is always the root of $T$,
hence for any two unrelated vertices $u,v$, whether $u$ and $v$ are
adjacent depends only on the label of $u$ and the label of $v$.

Now suppose $G$ admits a quasi-path $Q=(u_1,\ldots,u_\ell)$. A
{\em{block}} in $Q$ is a maximal contiguous subsequence of $Q$
consisting of pairwise related vertices. Thus $Q$ is partitioned into
blocks, say $B_1,\ldots,B_p$ appearing in this order on $Q$. Observe
that every block~$B_i$ either is a quasi-path itself or becomes a
quasi-path after removing its first vertex.  Since vertices of~$B_i$
are pairwise related, they are contained in an induced subgraph of $G$
that admits a tree model of height $h-1$ and using $m$ labels,
implying by the induction hypothesis that
\begin{equation}\label{eq:block-length}
  \textrm{every block has length at most $f(h-1,m)+1$.}
\end{equation}
Next, for every non-last block $B_i$ (i.e. $i\neq p$), let the
{\em{signature}} of $B_i$ be the following triple:
\begin{itemize}
\item the parity of the index of the last vertex of $B_i$,
\item the label of the last vertex of $B_i$, and
\item the label of its successor on $Q$, that is, the first vertex of
  $B_{i+1}$.
\end{itemize}
The next claim is the key step in the proof.

\begin{claim}\label{cl:seven-blocks}
There are no seven non-last blocks with the same signature.
\end{claim}
\begin{proof}\renewcommand{\qedsymbol}{\ensuremath{\lrcorner}}
  Supposing for the sake of contradiction that such seven non-last
  blocks exist, by taking the first, the fourth, and the seventh of
  them we find three non-last blocks $B_i,B_j,B_k$ with sames
  signature such that $1\leq i<j<k<p$ and $j-i>2$ and $k-j>2$.  Let
  $1\leq a<b<c<\ell$ be the indices on $Q$ of the last vertices of
  $B_i,B_j,B_k$, respectively.  By the assumption,
  $\lambda(u_a)=\lambda(u_b)=\lambda(u_c)$,
  $\lambda(u_{a+1})=\lambda(u_{b+1})=\lambda(u_{c+1})$, and $a,b,c$
  have the same parity.  Suppose for now that $a,b,c$ are all even;
  the second case will be analogous.  Further, the assumptions $j-i>2$
  and $k-j>2$ entail $b>a+2$ and $c>b+2$.

Observe that $u_{a+1}$ and~$u_b$ have to be related. Indeed, $u_{a}$
has the same label as~$u_{b}$, while it is unrelated and adjacent to
$u_{a+1}$.  So if $u_{a+1}$ and $u_b$ were unrelated, then they would
be adjacent as well, but this is a contradiction because $a+1$ is odd,
$b$ is even, and $a+2<b$.  Similarly $u_a$ and $u_{c+1}$ are related
and $u_b$ and $u_{c+1}$ are related. By transitivity we find
that~$u_b$ and $u_{b+1}$ are related, a contradiction.

The case when $a,b,c$ are all odd is analogous: we similarly find that
$u_a$ is related to~$u_{b+1}$, $u_a$ is related to $u_{c+1}$, and
$u_b$ is related to $u_{c+1}$, implying that $u_b$ is related to
$u_{b+1}$, a contradiction. This concludes the proof.
\end{proof}

Since there are $2m^2$ different signatures, \cref{cl:seven-blocks}
implies that
\begin{equation}\label{eq:block-number}
\textrm{the number of blocks is at most $12m^2+1$.}
\end{equation}
Assertions Equation~\ref{eq:block-length} and
Equation~\ref{eq:block-number} together imply that
$\ell\leq (f(h-1,m)+1)(12m^2+1)$.  As $Q$ was chosen arbitrarily, we
may set
$$f(h,m)\coloneqq (f(h-1,m)+1)\cdot (12m^2+1).$$
This concludes the proof of \cref{cl:conn-mod-quasi} and of
\cref{lem:exc-quasi}.
\end{proof}

Now \cref{lem:go-right} immediately follows from the following
(essentially reformulated) statement.

\begin{lemma}\label{lem:unambigious-neighbor}
  For every class $\CCC$ of graphs of bounded shrubdepth there exists
  a constant $p\in \N$ such that the following holds.  Suppose
  $G\in \CCC$ and $A$ and $B$ are two disjoint subsets of vertices of
  $G$ such that every vertex of $A$ has a neighbor in $B$.  Then there
  exist subsets $B_1,\ldots,B_p\subseteq B$ with the following
  property: for every vertex $v\in A$ there exists $i\in [p]$ such
  that $v$ has exactly one neighbor in $B_i$.
\end{lemma}
\begin{proof}
  Call a vertex $u\in B$ a {\em{private neighbor}} of a vertex
  $v\in A$ is $u$ is the only neighbor of~$v$ in $B$.  Consider the
  following procedure which iteratively removes vertices from $A$ and
  $B$ until $A$ becomes empty.  The procedure proceeds in rounds,
  where each round consists of two reduction steps, performed in
  order:

%
%
%
%

\begin{enumerate}
\item {\bf{$B$-reduction:}} As long as there exists a vertex $u\in B$
  that is not a private neighbor of any $v\in A$, remove $u$ from $B$.
\item {\bf{$A$-reduction:}} Remove all vertices from $A$ that have
  exactly one neighbor in $B$.
\end{enumerate}

Observe that in the $B$-reduction step we never remove any vertex that
is a private neighbor of some vertex in $A$, so during the procedure
we maintain the invariant that every vertex of~$A$ has at least one
neighbor in $B$.  Note also that in any round, after the $B$-reduction
step the set $B$ remains nonempty, due to the invariant, and then
every vertex of $B$ is a private neighbor of some vertex of $A$.
Thus, the $A$-reduction step will remove at least one vertex from $A$
per each vertex of $B$, so the size of $A$ decreases in each round.
Consequently, the procedure stops after a finite number of rounds, say
$\ell$, when $A$ becomes empty.

Let $B_1,\ldots,B_\ell$ be subsets of the original set $B$ such that
$B_i$ denotes $B$ after the $i$th round of the procedure.  Further,
let $A_1,\ldots,A_\ell$ be the subsets of the original set $A$ such
that $A_i$ comprises vertices removed from $A$ in the $i$th round.
Note that $A_1,\ldots,A_\ell$ form a partition of $A$.  The following
properties follow directly from the construction:
\begin{enumerate}
\item\label{p:1} Every vertex of $A_i$ has exactly one neighbor in
  $B_i$, for each $1\leq i\leq \ell$.
\item\label{p:2} Every vertex of $A_i$ has at least two neighbors in
  $B_{i-1}$, for each $2\leq i\leq \ell$.
\item\label{p:3} Every vertex of $B_i$ has at least one neighbor in
  $A_i$, for all $1\leq i\leq \ell$.
\end{enumerate}
For Property~\ref{p:2} observe that otherwise such a vertex would be
removed in the previous round.

Property \ref{p:1} implies that subsets $B_1,\ldots,B_\ell$ satisfy
the property requested in the lemma statement.  Hence, it suffices to
show that $\ell$, the number of rounds performed by the procedure, is
universally bounded by some constant $p$ depending on the class $\CCC$
only.

Take any vertex $v_\ell\in A_\ell$. By Property~\ref{p:1} and
Property~\ref{p:2}, it has at least two neighbors in~$B_{\ell-1}$, out
of which one, say $u_\ell$, belongs to $B_\ell$, and another, say
$u_{\ell-1}$, belongs to $B_{\ell-1}-B_\ell$. Next, by
Property~\ref{p:3} we have that $u_{\ell-1}$ has a neighbor
$v_{\ell-1}\in A_{\ell-1}$. Observe that $v_{\ell-1}$ cannot be
adjacent to $u_\ell$, because $v_{\ell-1}$ has exactly one neighbor in
$B_{\ell-1}$ by Property~\ref{p:1} and it is already adjacent to
$u_{\ell-1}\neq u_{\ell}$. Again, by Property~\ref{p:1} and
Property~\ref{p:2} we infer that $v_{\ell-1}$ has another neighbor
$u_{\ell-2}\in B_{\ell-2}-B_{\ell-1}$.  In turn, by Property~\ref{p:3}
again $u_{\ell-2}$ has a neighbor $v_{\ell-2}\in A_{\ell-2}$, which is
non-adjacent to both $u_{\ell-1}$ and $u_{\ell}$, because $u_{\ell-2}$
is its sole neighbor in~$B_{\ell-2}$.  Continuing in this manner we
find a sequence of vertices
$$(v_1,u_1,v_2,u_2,\ldots,v_\ell,u_\ell)$$
with the following properties: each two consecutive vertices in the
sequence are adjacent and for each $i<j$, $v_i$ is non-adjacent to
$u_j$. This is a quasi-path of length $2\ell$.  By
\cref{lem:exc-quasi}, there is a universal bound $q$ depending only on
$\CCC$ on the length of quasi-paths in $G$, implying that we may take
$p=\lfloor q/2\rfloor$.
\end{proof}

\section{Proof of \cref{lem:parallel}}

\begin{proof}[of \cref{lem:parallel}]
  It is enough to consider the case when $\interp I$ is an atomic
  operation.
  We assume that the input structure is a bundling $\bigcup \cal K^X$
  of $\cal K$, given by a function $f\from V(\bigcup \cal K)\to X$.
  Note that elements of $V(\bigcup \cal K)$ can be identified in the
  structure as those that are in the domain of $f$.
	
  Let $\sim$ be the equivalence relation on $V(\bigcup \cal K)$, where
  $x\sim y$ if and only if $f(x)=f(y)$.  Note that $\sim$ can be added
  to the structure by an extension operation parameterized by the
  formula $f(x)=f(y)$.  We now consider cases depending on what atomic
  operation~$\interp I$~is.
\begin{itemize}
\item If $\interp I$ is a reduct or restriction operation, then we set
  ${\interp I}^\star=\interp I$ (we may assume that a restriction does
  not remove elements of $X$ by appropriate relativization, so
  that~${\interp I}^\star$ indeed outputs a bundling).

\item If $\interp I$ is an extension operation parameterized by a
  quantifier-free formula $\phi(x_1,\ldots,x_k)$, then set
  ${\interp I}^\star$ to be the extension operation parameterized by
  the formula
  $\phi(x_1,\ldots,x_k)\land\bigwedge_{i,j\in [k]} (x_i\sim x_j)$.

\item If $\interp I$ is a function extension operation parameterized
  by a formula $\varphi(x,y)$, then set $\interp I^\star$ to be
  function extension operation parameterized by the formula
  $\varphi(x,y)\land (x\sim y)$.
  
\item If $\interp I$ is a copy operation, then $\interp I^\star$ is
  defined as the composition of a copy operation and a function
  extension operation that introduces a new function $f^\star$ in
  place of~$f$ defined as follows. We first define a function
  $\mathsf{origin}(x)$ as follows.  Recall that when copying, we
  introduce a new unary predicate, say $P$, marking the newly created
  vertices and each vertex is made adjacent to its new copy. We let
  $\mathsf{origin}(x)$ be defined by
  $\psi_{\mathsf{origin}}(x,y)\coloneqq P(x)\wedge E(x,y)$. We now
  define $f^\star(x)=f(\mathsf{origin}(x))$. The resulting bundling is
  given by the function $f^\star$.

\item If $\interp I$ is a unary lift, say parameterized by a function
  $\sigma$, then set ${\interp I}^\star$ to be the unary lift
  parameterized by the function $\sigma^\star$ that applies $\sigma$
  to each structure from $\cal K$ separately, investigates all
  possible ways of picking one output for each structure in $\cal K$,
  and returns the set of bundlings of sets formed in this
  way.
\end{itemize}
\end{proof}

%% file: qe-new.tex
\section{Quantifier elimination}\label{app:qe}

In this section we provide the missing proofs of the lemmas from
\cref{sec:qe}.

  \subsection{Proof of \cref{lem:star}}
  
\begin{proof}[of \cref{lem:star}]
  We show that if $\CCC$ is a class of graphs of bounded expansion,
  $G\in \CCC$ and $f\from V(G)\partto V(G)$ is a partial function that
  is guarded by $G$, then~$f$ is $\ell$-guidable, for some $\ell$
  depending only on $\CCC$. Then the claim of the lemma follows by
  \cref{lem:local}.
  
  First, consider the special case when $\CCC$ is a class of treedepth
  $h$, for some $h\in\N$.  For each $G\in \CCC$, fix a forest $F$ of
  depth $h$ with $V(F)=V(G)$ such that every edge in~$G$ connects
  comparable nodes of $F$. Label every vertex $v$ of $G$ by the depth
  of $v$ in the forest $F$, using labels $\set{1,\ldots,h}$. It is
  easy to see that the corresponding partition of $V(G)$ is a guidance
  system of order $h$ for $f$.

  Now the general case, when $\CCC$ is a class which has a $2$-cover
  $\cal U$ of bounded treedepth. Let
  $N=\sup\set{|\cal U_G|:G\in \CCC}$, and let $h$ be the treedepth of
  the class $\CCC[\cal U]$. Let $G\in \CCC$ be a graph and let
  $f\from V(G)\to V(G)$ be a function which is guarded by $G$.  Then
  $f|_U$ is $h$-guidable by the previous case, and hence $f$ is
  $(h\cdot N)$-guidable by \cref{lem:glue-directable}.
\end{proof}

\subsection{Proof of \cref{lem:qe-trees0}: quantifier elimination on trees of bounded depth}\label{sec:automata}
We first give a quantifier elimination procedure for colored trees of
bounded depth.  In the following, we consider {\em{$\Sigma$-labeled
    trees}}, that is, unordered rooted trees $t$ where each node is
labeled with exactly one element of $\Sigma$. We write $t(v)$ for the
label of a node $v$ in the tree $t$.  In this section we model trees
by their parent functions, that is, we consider them as structures
where the universe of the structure is the node set, there is a unary
relation for each {\em{label}} from $\Sigma$, and there is one partial
function that maps each node to its parent (the roots are not in the
domain).  A $\Gamma$-relabeling of a $\Sigma$-labeled tree~$t$ is any
$\Gamma$-labeled tree whose underlying unlabeled tree is the same as
that of $t$.  As usual, a class of trees $\TTT$ has {\em{bounded
    height}} if there exists $h\in \N$ such that each tree in $\TTT$
has height at most $h$.

For convenience we now regard sets of free variables of formulas,
instead of traditional tuples.  That is, if $\varphi$ is a formula
with free variables $X$ and $\nu\colon X\to V(t)$ is a valuation of
variables from $X$ in a tree $t$, then we write $t,\nu\models \varphi$
if the formula $\varphi$ is satisfied in $t$ when its free variables
are evaluated as prescribed by $\nu$.

Our quantifier elimination procedure is provided by the
following~lemma, which implies \cref{lem:qe-trees0}.

\begin{lemma}\label{le:qe-forests}
  Let $\TTT$ be a class of $\Sigma$-labeled trees of bounded height
  and let $\varphi$ be a first-order formula over the signature of
  $\Sigma$-labeled trees with free variables $X$.  Then there exists a
  finite set of labels $\Gamma$, a $\Gamma$-relabeling $\wh t$ of $t$,
  and a quantifier-free formula $\wh \phi$ over the signature of
  $\Gamma$-labeled trees with free variables $X$, such that for each
  valuation $\nu$ of $X$ in $t$ we have
$$t,\nu\models \phi\qquad\textrm{if and only if}\qquad \wh t,\nu\models \wh \phi.$$
\end{lemma}

The result immediately lifts to classes of forests of bounded depth,
which are modeled the same way as trees, i.e., using a unary parent
function.
\begin{corollary}\label{cor:qe-forests}
  The same statement as above holds for a class $\FFF$ of
  $\Sigma$-labeled forests of bounded height and a first-order formula
  $\psi$ over the signature $\Sigma$-labeled forests.
\end{corollary}
\begin{proof}
  Let $\FFF$ be a class of $\Sigma$-labeled forests of bounded height
  and let $\psi$ be a first-order formula with free variables
  $X$. Construct a class of $\Sigma$-labeled trees $\TTT$, by
  prepending an unlabeled root $r_f$ to each forest $f$ in $\FFF$,
  yielding a tree $t_f$. We may rewrite the formula~$\psi$ to a
  first-order formula $\phi$ such that $f,\nu\models \psi$ if and only
  if $t_f,\nu\models\phi$, for every $f\in \FFF$ and every valuation
  $\nu$ of $X$ in $f$.

  Apply \cref{le:qe-forests} to $\TTT$, yielding a relabeling $\wh t$
  of each tree $t$ in $\TTT$, using some finite set of labels
  $\Gamma$.  This relabeling yields a relabeling $\wh f$ of each
  forest $f\in \FFF$, where each non-root node $v$ is labeled by a
  pair of labels: the label of $v$ in the tree $\wh t_f$, and the
  label of the root of $\wh t_f$. Furthermore, we have
  $t_f,\nu\models\phi$ if and only if $\wh t_f,\nu\models \wh \phi$,
  for every valuation $\nu$.  Note that all quantifier-free properties
  involving the prepended root $r_f$ in the $\Gamma$-labeled tree
  $\wh t_f$ can be decoded from the labeled forest $\wh f$: the unary
  predicates that hold in $r_f$ are encoded in all the vertices of
  $\wh f$, and $r_f$ is the parent of the roots of $\wh f$ (the
  elements for which the parent function is undefined).  It follows
  that we may rewrite the formula $\wh \phi$ to a formula $\wh\psi$
  such that $\wh t_f,\nu\models \wh\phi$ if and only if
  $\wh f,\nu\models \wh\psi$, for every valuation $\nu$ of $X$ in
  $f$. Reassuming, $f,\nu\models \psi$ if and only if
  $\wh f,\nu\models\wh \psi$, for every $f\in \FFF$ and every
  valuation $\nu$ of $X$ in $f$.
\end{proof}
\cref{cor:qe-forests} immediately implies \cref{lem:qe-trees0}. It
remains to prove \cref{le:qe-forests}. Before proving
\cref{le:qe-forests}, we recall some standard automata-theoretic
techniques.

\bigskip

\newcommand{\res}[2]{#2\!\downharpoonright_{\;#1}} We define tree
automata which process unordered labeled trees. Such automata process
an input tree $t$ from the leaves to the root
assigning states to each node in the tree.  The state assigned to the
current node $v$ depends only on the label $t(v)$ and the multiset of
states labeling the children of $v$, where the multiplicities are
counted only up to a certain fixed threshold. Because of that, we call
these automata \emph{threshold tree automata}.

\newcommand{\aut}[1]{\mathcal{#1}} 

We develop all the simple facts about tree automata needed for our
purposes below. We refer to~\cite{DBLP:books/ws/automata2012/Loding12}
for a general introduction. Note that what is usually considered under
the notion of \emph{tree automata} are automata which process
\emph{ordered} trees, i.e., trees where the children of each node are
ordered. Tree automata collapse in expressive power to threshold tree
automata in the case when they are required to be independent of the
order, i.e., if $\aut A$ is a tree automaton with the property that
for any two ordered trees~$t,t'$ which are isomorphic as unordered
trees, either both $t$ and $t'$ are accepted by $\aut A$ or both~$t$
and $t'$ are rejected by $\aut A$, then the language (i.e., set) of
trees accepted by $\aut A$ is equal to the language of trees accepted
by some threshold automaton. Therefore, the theory of threshold tree
automata is a very simple and special case of that of tree
automata. We now recall some simple facts about such automata.
 
\bigskip

Fix a set of labels $Q$. A $Q$-multiset is a multiset of elements of
$Q$. If $\tau$ is a number and $X$ is a $Q$-multiset, then by
$\res\tau X$ we denote the maximal multiset $X'\subset X$ where the
multiplicity of each element is at most $\tau$.  In other words, for
every element whose multiplicity in $X$ is more than $\tau$, we put it
exactly $\tau$ times to $X'$; all the other elements retain their
multiplicities.

\medskip
  
We define \emph{threshold tree automata} as follows.  A threshold tree
automaton is a tuple $(\Sigma, Q,\tau, \delta, F)$, consisting~of
\begin{itemize}
  \item a finite input alphabet $\Sigma$;
  \item a finite state space $Q$;
  \item a \emph{threshold} $\tau\in \N$;
  \item a transition relation $\delta$, which is a finite set of rules
    of the form $(a,X,q)$, where $a\in \Sigma$, $q\in Q$, and $X$ is a
    $Q$-multiset in which each element occurs at most $\tau$ times;
    and
  \item an \emph{accepting condition} $F$, which is a subset of $Q$.
\end{itemize}

A \emph{run} of such an automaton over a $\Sigma$-labeled tree $t$ is
a $Q$-labeling $\rho:V(t)\to Q$ of~$t$ satisfying the following
condition for every node $x$ of $t$:
\begin{quote}
  If $t(x)=a, \rho(x)=q$ and $X$ is the multiset of the $Q$-labels of
  the children of~$x$ in $t$, then $(a,\res\tau X,q)\in \delta$.
\end{quote}

The automaton \emph{accepts} a $\Sigma$-labeled tree $t$ if it has a
run $\rho$ on $t$ such that $\rho(r)\in F$, where $r$ is the root of
$t$.  The \emph{language} of a threshold tree automaton is the set of
$\Sigma$-labeled trees it accepts.  A language $L$ of $\Sigma$-labeled
trees is {\em{threshold-regular}} if there is a threshold tree
automaton whose language is $L$; we also say that this automaton
{\em{recognizes}} $L$.

An automaton is \emph{deterministic} if for all $a\in \Sigma$ and all
$Q$-multisets $X$ in which each element occurs at most $\tau$ times
there exists $q$ such that $(a,X,q)\in \delta$ and whenever
$(a,X,q),(a,X,q')\in \delta$, then $q=q'$.  Note that a deterministic
automaton has a unique run on every input tree.

The next lemma explains basic properties of threshold tree automata
and follows from standard automata constructions.  In the lemma we
speak about {\em{monadic second-order logic}} (MSO), which is the
extension of first-order logic by quantification over unary
predicates.

\pagebreak
\begin{lemma}\label{lem:automata}
The following assertions hold:
\begin{enumerate}[(1)]
\item\label{determinisation} For every threshold automaton there is a
  deterministic threshold automaton with the same language.
\item\label{closure-bc} Threshold-regular languages are closed under
  boolean operations.
\item\label{closure-relabel} If $f\colon \Sigma\to \Gamma$ is any
  function and $L$ is a threshold-regular language of $\Sigma$-labeled
  trees, then the language $f(L)$ comprising trees obtained from trees
  of $L$ by replacing each label by its image under $f$ is also
  threshold-regular.
\item\label{formula-to-automata} For every MSO sentence $\phi$ in the
  language of\, $\Sigma$-labeled trees there is a deterministic
  threshold automaton $\aut A_\phi$ whose language is the set of trees
  satisfying~$\phi$.
\end{enumerate}
\end{lemma}
\begin{proof}
  Assertion \ref{determinisation} follows by applying the standard
  powerset determinization construction.  For assertion
  \ref{closure-bc}, it follows from \ref{determinisation} that every
  threshold-regular language is recognized by a deterministic
  threshold tree automaton.  Then, for conjunctions we may use the
  standard product construction and for negation we may negate the
  accepting condition.  For assertion \ref{closure-relabel}, an
  automaton recognizing $f(L)$ can be constructed from an automaton
  recognizing $L$ by nondeterministically guessing labels from
  $\Sigma$ consistently with the given labels from $\Gamma$, so that
  the guessed $\Sigma$-labeling is accepted by the automaton
  recognizing $L$.  Now assertion \ref{formula-to-automata} follows
  from~\ref{determinisation},~\ref{closure-bc},
  and~\ref{closure-relabel} in a standard way, because every MSO
  formula can be constructed from atomic formulas using boolean
  combinations and existential quantification (which can be regarded
  as a relabeling $f$ that forgets the information about the
  quantified set).
\end{proof}


Let $X$ be a finite set of (first-order) variables and let
$\Sigma_X=\Sigma\times \pow{X}$.
Given a tree $t$ and a partial valuation $\nu\from X\partto V(t)$, let
$t\otimes \nu$ be the $\Sigma_X$-tree obtained from $t$, by replacing,
for each node $u$ of $t$, the label $a$ of $u$ by the pair $(a,Y)$
where $Y=\nu^{-1}(u)\subseteq X$.

Toward the proof of \cref{le:qe-forests}, consider a first-order
formula $\phi$ over $\Sigma$-labeled trees with free variables $X$.
We can easily rewrite $\varphi$ to a first-order sentence $\psi$ over
$\Sigma_X$-labeled trees such that $t,\nu\models \phi$ if and only if
$t\otimes \nu \models \psi$ for every $\Sigma$-labeled tree~$t$ and
valuation $\nu\colon X\to V(t)$.  By
\cref{lem:automata}\ref{formula-to-automata} there is a deterministic
threshold automaton $\aut A_{\psi}$ whose language is exactly the set
of $\Sigma_X$-labeled trees satisfying $\psi$.

Denote by $Q$ the set of states and by $K$ the threshold of
$\aut A_{\psi}$, and let $M=K+|X|$.  Denote by $\Delta$ the set of
$Q$-multisets in which every element occurs at most $M$ times.

Given a $\Sigma$-labeled tree $t$ and a partial valuation
$\nu\from X\partto V(t)$, define $\rho_\nu$ as the $Q$-labeling of $t$
which is the unique run of $\aut A_{\psi}$ over $t\otimes \nu$.  For a
node $u$ of $t$, let $C_\nu(u)$ be the $Q$-multiset defined as
follows:
    \begin{align*}    C_\nu(u)&=\set{\rho_\nu(w):\textit{ $w$ is a child of $u$ in $t$}}.
\end{align*}

Define a new set of labels $\Gamma=\Sigma\times \Delta$, and a
$\Gamma$-relabeling $\wh t$ of $t$ as follows: for each $u\in V(t)$,
say with label $a\in \Sigma$ in $t$, the label of $u$ in $\wh t$ is
the pair $(a,\res M{C_\emptyset(u)})$, where~$\emptyset$ is the
partial valuation that leaves all variables of $X$ unassigned.  Our
goal now is to prove that this relabeling $\wh t$ of $t$ satisfies the
conditions expressed in \cref{le:qe-forests}.  To this end, given a
valuation $\nu$ of $X$ in $\wh t$, let $\wh t|_\nu$ denote the
$\Gamma_X$-labeled tree obtained from $\wh t\otimes \nu$ by
restricting the node set to the set of ancestors of nodes in the image
$\nu(X)$ of $\nu$.

\begin{lemma}\label{cor:set-defines}
  There is a set of $\Gamma_X$-labeled trees $\cal R$ such that for
  every $\Sigma$-labeled tree $t$ and valuation $\nu$ of $X$
  in~$t$,
  $$t,\nu\models \phi\textit{\quad if and only if\quad}\wh t|_\nu\in
  \cal R.$$
\end{lemma}
\begin{proof}
  Fix a tree $t$ and a valuation $\nu$ of $X$ in $t$. We say that a
  node $u$ of $t$ is \emph{nonempty} if it has a descendant which is
  in the image of $\nu$.  For node $u$ of $t$ define the following
  $Q$-multisets:
  \begin{align*}    N_\emptyset(u)&=\set{\rho_\emptyset(w):\textit{ $w$ is a nonempty child of $u$}},\\
    N_\nu(u)&=\set{\rho_\nu(w):\textit{
              $w$ is a nonempty child of $u$}}.
  \end{align*}
  Note that since there are at most $|X|$ nonempty children of a given
  node $u$, there is a finite set $Z$ independent of $t$ and $\nu$
  such that the functions $N_\nu$ and $N_\emptyset$ take values in
  $Z$. Fix a node $u$ of $t$.
	
  \begin{claim}\label{cl:depends}
    The state $\rho_\nu(u)$ is uniquely determined by the label of $u$
    in $t\otimes\nu$, and the $Q$-multisets
    $\res M{C_\emptyset(u)},N_\emptyset(u)$ and $N_\nu(u)$, i.e., 
    there is a function
    $f\from \Sigma_X\times \Delta\times Z\times Z\to Q$ such that for
    every tree $t$, valuation $\nu$ and node $u$,
    \begin{align}\label{eq:dyn}
      \rho_\nu(u)=f(\,\mathrm{label\ of }\ u\ \mathrm{ in }\ t\otimes\nu\, , \,\res M{C_\emptyset(u)}\, ,\, N_\emptyset(u)\, ,\, N_\nu(u)\, ). 
    \end{align}     
  \end{claim}
  \begin{proof}
    Clearly $N_\emptyset(u)\subset C_\emptyset(u)$, as multisets.
    Moreover, the following equality among multisets holds:
    \begin{align}\label{eq:multiset}
      C_\nu(u)&=(C_\emptyset(u)-N_\emptyset(u))+N_\nu(u).
    \end{align}
    This is because the automaton $\aut A_\psi$ is deterministic and
    therefore $\rho_\nu(w)=\rho_\emptyset(w)$ for all nodes $w$ which
    are not nonempty.
    From~Equation~\ref{eq:multiset}, the fact that $N_\emptyset(u)$
    has at most~$|X|$ elements and $M=K+|X|$, it follows that
    \begin{align}\label{eq:multiset-res}
      \res K{((\res M {C_\emptyset(u)}-N_\emptyset(u))+N_\nu(u))}=\res K {(C_\nu(u))}.  
    \end{align}
    By definition of the run of $\aut A_\psi$ on $t\otimes\nu$, the
    state $\rho_\nu(u)$ is determined by the label of $u$ in
    $t\otimes \nu$ and by $\res K {(C_\nu(u))}$.  It follows
    from~Equation~\ref{eq:multiset-res} that $\rho_\nu(u)$ is uniquely
    determined by the label of $u$ in $t\otimes \nu$,
    $\res M{(C_\emptyset(u))}$, and the $Q$-multisets $N_\emptyset(u)$
    and $N_\nu(u)$, proving the claim.
\end{proof}

From \cref{cl:depends} it follows that the state $\rho_\nu(r)$, where
$r$ is the root of $t$, depends only on the tree $\wh t|_\nu$.
Indeed, we can inductively compute the states $\rho_\nu(u)$ and
$\rho_\emptyset(u)$, moving from the leaves of $\wh t|_\nu$ towards
the root, as follows.  Suppose $u$ is a node of $\wh t|_\nu$ such
that~$\rho_\nu(v)$ and~$\rho_\emptyset(v)$ have been computed for all
the nonempty children $v$ of $u$ (in particular, this holds if $u$ is
a leaf of $\wh t|_\nu$).  Then, we can determine the
multisets~$N_\nu(u)$ and $N_\emptyset(u)$ using their definitions, and
consequently, we can determine~$\rho_\nu(u)$ by~Equation~\ref{eq:dyn},
whereas~$\rho_\emptyset(u)$ only depends on $\res K{C_\emptyset(u)}$
and on the label of $u$ in $t$.  Note that both the label of $u$ in
$t$ and the multiset $\res K{C_\emptyset(u)}$ are encoded in the label
of $u$ in $\wh t$.

As shown above, for any tree $t$ and valuation $\nu$, the state of
$\rho_\nu$ at the root depends only on $\wh t|_\nu$.  On the other
hand, $t,\nu\models\phi$ if and only if the state of $\rho_\nu(r)$ at
the root is an accepting state.  Hence, whether or not
$t,\nu\models \phi$, depends only on the tree $\wh t|_\nu$. This
proves the lemma.
\end{proof}

Finally, we observe the following.

\begin{lemma}\label{lem:induced-forest-check}
  For each $\Gamma_X$-labeled tree $s$ there exists a quantifier-free
  formula $\psi_s$ over the signature of $\Gamma$-labeled trees with
  free variables $X$ such that the following holds: for every
  $\Gamma$-labeled tree $t$ and valuation $\nu$ of $X$ in $t$, we have
$$t,\nu\models \psi_s\qquad\textrm{if and only if}\qquad\textrm{$t|_\nu$ is isomorphic to $s$.}$$
\end{lemma}
\begin{proof}
  Observe that the ancestors of nodes in $\nu(X)$ may be obtained by
  applying the parent function to them.  Thus, using a quantifier-free
  formula we may check whether each node of $\nu(X)$ lies at depth as
  prescribed by $s$, whether its ancestors have labels as prescribed
  by $s$, and whether the depth of the least common ancestor of every
  pair of nodes of $\nu(X)$ is as prescribed by $s$.  Then $t|_\nu$ is
  isomorphic to $s$ if and only if all these conditions hold.
\end{proof}

With all the tools prepared, we may prove \cref{le:qe-forests}.

\begin{proof}[of \cref{le:qe-forests}]
  Let $\cal R_h$ be the intersection of $\cal R$ with the class of
  trees of height at most $h$.  Since each tree from $\cal R$ has at
  most $|X|$ leaves by definition, $\cal R_h$ is finite and its size
  depends only on $|X|$ and $h$.  By \cref{cor:set-defines}, it now
  suffices to define~$\wh \phi$ as the disjunction of formulas
  $\psi_s$ provided by \cref{lem:induced-forest-check} over
  $s\in \cal R_h$.
\end{proof}